\newcommand\fullversion

\ifcsname fullversion\endcsname

\documentclass[runningheads,envcountsame]{article}
\usepackage[a4paper, total={6.8in, 8in}]{geometry}
\usepackage{amsthm}

\newtheorem{theorem}{Theorem}
\newtheorem{lemma}[theorem]{Lemma}
\newtheorem{example}[theorem]{Example}
\newtheorem{definition}[theorem]{Definition}
\newtheorem{proposition}[theorem]{Proposition}
\newtheorem{claim}[theorem]{Claim}

\else

\documentclass[runningheads,envcountsame]{llncs}

\fi

\usepackage[T1]{fontenc}
\usepackage{amsmath,amssymb}
\usepackage{enumerate}
\usepackage{graphicx}
\usepackage{tikz}
\usetikzlibrary{trees}
\usetikzlibrary{shapes}
\usetikzlibrary{fit}
\usetikzlibrary{shadows}
\usetikzlibrary{backgrounds}
\usetikzlibrary{arrows,automata}

\ifcsname fullversion\endcsname
\usepackage[bibliography=common]{apxproof}
\else
\usepackage[bibliography=common,appendix=strip]{apxproof}
\fi

\newtheoremrep{lemma}[theorem]{Lemma}
\newtheoremrep{corollary}[theorem]{Corollary}

\usepackage{mathtools}
\usepackage{amsmath}
\usepackage{amssymb}
\usepackage{algorithm}
\usepackage{algpseudocode}
\usepackage{enumitem}
\usepackage{xparse}
\usepackage{enumerate}
\usepackage{xspace}
\usepackage{algorithm}
\usepackage{algpseudocode}
\usepackage{float}
\usepackage{xcolor}
\usepackage{soul}
\usepackage{thm-restate}
\usepackage{pdfpages}
\usepackage{hyperref}
\usepackage{wrapfig}
\usepackage{tikz}
\usetikzlibrary{patterns}
\usetikzlibrary{fit}
\usetikzlibrary{calc}
\usetikzlibrary{decorations.pathreplacing}
\usetikzlibrary{decorations.text}
\usetikzlibrary{petri,positioning}
\usepackage{subcaption}

\definecolor{mygreen}{RGB}{36, 200, 100}
\definecolor{myyellow}{RGB}{220, 220, 30}
\newcommand{\mybox}[1]{\colorbox{myyellow!40}{$\displaystyle #1$}}

\newcommand{\mathcommand}[2]{\newcommand{#1}{\ensuremath{#2}\xspace}}

\mathcommand{\FO}{\textsc{FO}}
\mathcommand{\prefFO}{\textsc{FO}^{\textsc{pref}}}
\mathcommand{\prefFOdw}{\prefFO[\lesssim]}  
\mathcommand{\FOtwo}{\FO^2}
\mathcommand{\FOthree}{\FO^3}
\mathcommand{\MSO}{\textsc{MSO}}
\NewDocumentCommand{\FOeq}{O{}}{\ensuremath{\FO_{#1}[\sim]}\xspace}
\NewDocumentCommand{\FOtwoeq}{O{}}{\ensuremath{\FOtwo_{#1}[\sim]}\xspace}

\NewDocumentCommand{\FOeqord}{}{\ensuremath{\FO[\sim,<]}\xspace}
\NewDocumentCommand{\FOtwoeqord}{}{\ensuremath{\FO^{2}[\sim,<]}\xspace}

\NewDocumentCommand{\MSOordeq}{O{}}{\ensuremath{\MSO^{#1}[\lesssim]}\xspace}

\mathcommand{\plusone}{+1}

\NewDocumentCommand{\PARTSYNTH}{O{\FOeq}}{\ensuremath{\textsc{PartSynth}(#1)}\xspace}
\NewDocumentCommand{\SHAREDSYNTH}{O{\FOeq}}{\ensuremath{\textsc{SharedSynth}(#1)}\xspace}

\NewDocumentCommand{\foeq}{O{k}}{\ensuremath{\equiv^{\FO}_{#1}}\xspace}
\NewDocumentCommand{\fotwoeq}{O{k}}{\ensuremath{\equiv^{\FOtwo}_{#1}}\xspace}   
\NewDocumentCommand{\msoeq}{O{k}}{\ensuremath{\equiv^{\MSO}_{#1}}\xspace}
\NewDocumentCommand{\prefFOeq}{O{k}}{\ensuremath{\equiv^{\prefFO}_{#1}}\xspace}
\NewDocumentCommand{\prefFOdweq}{O{k}}{\ensuremath{\equiv^{\prefFOdw}_{#1}}\xspace}
\NewDocumentCommand{\msoordeqeq}{O{k}}{\ensuremath{\equiv^{\MSOordeq}_{#1}}\xspace}

\mathcommand{\PTIME}{\textsc{PTime}}
\mathcommand{\LOGSPACE}{\textsc{LogSpace}}
\mathcommand{\PSPACE}{\textsc{PSpace}}
\mathcommand{\ACO}{\textsc{AC}\ensuremath{^0}}

\mathcommand{\Alp}{A}
\mathcommand{\Proc}{\mathbb{P}}
\mathcommand{\dword}{\mathfrak{w}}
\mathcommand{\dwordbis}{\mathfrak{\widehat w}}
\mathcommand{\dwordnew}{\mathfrak{w}'}
\mathcommand{\dwordbisnew}{\mathfrak{\widehat w}'}

\mathcommand{\procex}{\exists^{\text{proc}}}
\mathcommand{\boundex}{\exists^{\text{bnd}}}
\mathcommand{\prefex}{\exists^{\text{pref}}}

\NewDocumentCommand{\pebble}{O{i}}{\ensuremath{p_{#1}}\xspace}
\NewDocumentCommand{\pebblebis}{O{i}}{\ensuremath{\widehat{p}_{#1}}\xspace}
\NewDocumentCommand{\procpebble}{O{i}}{\ensuremath{p^{\text{proc}}_{#1}}\xspace}
\NewDocumentCommand{\procpebblebis}{O{i}}{\ensuremath{\widehat{p}^{\text{proc}}_{#1}}\xspace}
\NewDocumentCommand{\boundpebble}{O{i}}{\ensuremath{p^{\text{bnd}}_{#1}}\xspace}
\NewDocumentCommand{\boundpebblebis}{O{i}}{\ensuremath{\widehat{p}^{\text{bnd}}_{#1}}\xspace}
\NewDocumentCommand{\prefpebble}{O{i}}{\ensuremath{p^{\text{pref}}_{#1}}\xspace}
\NewDocumentCommand{\prefpebblebis}{O{i}}{\ensuremath{\widehat{p}^{\text{pref}}_{#1}}\xspace}

\NewDocumentCommand{\prefEFeq}{O{k}}{\ensuremath{\equiv^{\text{pref-EF}}_{#1}}\xspace}

\mathcommand{\Const}{C}
\mathcommand{\procConst}{\Const^{\text{proc}}}
\mathcommand{\boundConst}{\Const^{\text{bnd}}}
\mathcommand{\prefConst}{\Const^{\text{pref}}}
\NewDocumentCommand{\boundconst}{O{i}}{\ensuremath{c^{\text{bnd}}_{#1}}\xspace}
\NewDocumentCommand{\prefconst}{O{i}}{\ensuremath{c^{\text{pref}}_{#1}}\xspace}
\NewDocumentCommand{\procconst}{O{i}}{\ensuremath{c^{\text{proc}}_{#1}}\xspace}
\mathcommand{\boundl}{l^{\text{bnd}}}
\mathcommand{\prefl}{l^{\text{pref}}}
\mathcommand{\procl}{l^{\text{proc}}}

\NewDocumentCommand{\boundvar}{O{}}{\ensuremath{x^{\text{bnd}}_{#1}}\xspace}
\NewDocumentCommand{\prefvar}{O{}}{\ensuremath{x^{\text{pref}}_{#1}}\xspace}
\NewDocumentCommand{\procvar}{O{}}{\ensuremath{x^{\text{proc}}_{#1}}\xspace}

\mathcommand{\Var}{X}
\mathcommand{\procVar}{\Var^{\text{proc}}}
\mathcommand{\boundVar}{\Var^{\text{bnd}}}
\mathcommand{\prefVar}{\Var^{\text{pref}}}

\mathcommand{\move}{\alpha}
\mathcommand{\movebis}{\widehat \move}
\mathcommand{\movetwo}{\beta}
\mathcommand{\movetwobis}{\widehat \movetwo}

\newcommand{\Sys}{System\xspace}
\newcommand{\Envi}{Environment\xspace}

\mathcommand{\sig}{A}
\mathcommand{\sigE}{\sig_E}
\mathcommand{\sigS}{\sig_S}

\newcommand{\sys}{\mathrm{S}}
\newcommand{\env}{\mathrm{E}}
\mathcommand{\AlpE}{\Alp_\env}
\mathcommand{\AlpS}{\Alp_\sys}
\mathcommand{\ProcE}{\Proc_\env}
\mathcommand{\ProcS}{\Proc_\sys}
\mathcommand{\ProcM}{\Proc_{\sys\env}}
\NewDocumentCommand{\DW}{O{\vocab}O{\Proc}}{\ensuremath{\mathrm{DW}_{#1}^{#2}}}

\mathcommand{\arena}{\mathcal{A}}
\mathcommand{\Transi}{\Delta}
\mathcommand{\Acc}{\alpha_\formule}
\mathcommand{\counting}{\kappa}
\mathcommand{\tok}{t}
\mathcommand{\Tok}{\mathbb{T}}
\mathcommand{\TokE}{\Tok_{\env}}
\mathcommand{\TokS}{\Tok_{\sys}}

\NewDocumentCommand{\Plays}{O{\Tok}O{k}}{\ensuremath{\mathrm{Plays}_{#1}^{#2}}\xspace}
\mathcommand{\play}{\pi}
\mathcommand{\playbis}{\play'}        
\mathcommand{\playplus}{\play_{+}}
\mathcommand{\playbisplus}{\playbis_{+}}
\NewDocumentCommand{\partplay}{O{i}}{\ensuremath{\play_{|#1}}\xspace}
\mathcommand{\nextplay}{\bar\play}
\mathcommand{\nextplayplus}{\nextplay_+}

\NewDocumentCommand{\TypesFO}{O{k}}{\ensuremath{\mathrm{Types}_{\FO}^{#1}}\xspace}
\NewDocumentCommand{\Types}{O{k}}{\ensuremath{\mathrm{Types}_{\prefFO}^{#1}}\xspace}
\NewDocumentCommand{\tpword}{O{w}O{\prefFO}O{k}}{\ensuremath{\langle {#1} \rangle_{#2}^{#3}}\xspace}
\NewDocumentCommand{\fotpword}{O{w}O{k}}{\tpword[#1][\FO][#2]}
\NewDocumentCommand{\preffotpword}{O{w}O{k}}{\tpword[#1][\prefFO][#2]}
\NewDocumentCommand{\smallpreffotpword}{O{w}O{k}}{\tiny\langle {#1}\rangle_{\scriptscriptstyle\textsc{FO}^{\scriptscriptstyle\textsc{pref}}}^{\scriptscriptstyle #2}\xspace}

\mathcommand{\conf}{C}
\mathcommand{\config}{\conf}
\mathcommand{\confinit}{\conf_0}
\mathcommand{\configinit}{\confinit}

\mathcommand{\movegame}{m}
\mathcommand{\moveS}{\movegame_\sys}
\mathcommand{\moveE}{\movegame_\env}
\mathcommand{\nextMoveE}{\moveE^{2l}}

\mathcommand{\newgame}{(\nS,\nE)\text{-game}}
\mathcommand{\oldgame}{(\nS+1,\nE)\text{-game}}

\NewDocumentCommand{\minind}{O{\formule}}{\ensuremath{n^\text{min}_\text{E}}\xspace}
\NewDocumentCommand{\cut}{O{\formule}}{\ensuremath{n^\text{min}_\text{S}}\xspace}
\mathcommand{\pos}{l}
\mathcommand{\posinit}{l_0}
\mathcommand{\posbis}{\pos'}
\mathcommand{\spel}{\hat \pos}
\NewDocumentCommand{\afterE}{O{\pos}}{E^*(#1)}
\NewDocumentCommand{\succE}{O{\pos}}{E^{+1}(#1)}
\mathcommand{\confbis}{\widehat\conf}
\mathcommand{\confS}{\conf_\text{S}}
\mathcommand{\confbisS}{\confbis_\text{S}}
\mathcommand{\confE}{\conf_\text{E}}
\mathcommand{\confbisE}{\confbis_\text{E}}
\NewDocumentCommand{\numS}{O{\pos}}{\ensuremath{\confS(#1)}\xspace}
\NewDocumentCommand{\numE}{O{\pos}}{\ensuremath{\confE(#1)}\xspace}
\NewDocumentCommand{\numafterE}{O{\pos}}{\confE^*(#1)}
\NewDocumentCommand{\pot}{O{\pos}}{\rho(#1)}
\mathcommand{\victory}{\mathcal F}
\mathcommand{\game}{\mathfrak G}

\mathcommand{\vocabE}{\vocab_E}
\mathcommand{\vocabS}{\vocab_S}

\newcommand{\nE}{\ensuremath{n_E}\xspace}
\newcommand{\nS}{\ensuremath{n_S}\xspace}

\NewDocumentCommand{\fE}{O{k}}{\ensuremath{f_E(#1)}\xspace}
\NewDocumentCommand{\fS}{O{k}O{\nE}}{\ensuremath{f_S(#1,#2)}\xspace}

\mathcommand{\States}{\mathcal Q}
\mathcommand{\state}{q}
\mathcommand{\statebis}{\state'} 
\mathcommand{\istate}{\state_0}
\mathcommand{\stateone}{\state_1}
\mathcommand{\statetwo}{\state_2}
\mathcommand{\hstate}{\state_h}
\mathcommand{\Trans}{\mathcal T}
\mathcommand{\trans}{t}
\mathcommand{\transzero}{\trans_0}
\mathcommand{\transone}{\trans_1}
\mathcommand{\transtwo}{\trans_2}
\mathcommand{\transthree}{\trans_3}
\mathcommand{\mm}{M}
\mathcommand{\formmm}{\formule_{\mm}}

\NewDocumentCommand{\transition}{O{\trans}O{\state}O{\state}O{\cnt++}}{\ensuremath{#1:#2\xrightarrow{#4}#3}\xspace}

\newcommand{\strinc}{++}
\newcommand{\strdec}{--}
\newcommand{\strzero}{==0}

\NewDocumentCommand{\Transinc}{O{i}}{\ensuremath{\Trans^{\strinc}_{#1}}\xspace}
\NewDocumentCommand{\Transdec}{O{i}}{\ensuremath{\Trans^{\strdec}_{#1}}\xspace}
\NewDocumentCommand{\Transzero}{O{i}}{\ensuremath{\Trans^{\strzero}_{#1}}\xspace}

\NewDocumentCommand{\good}{O{}O{}O{x}}{\ensuremath{\text{Valid}_{#1}^{#2}(#3)}\xspace}
\NewDocumentCommand{\goode}{O{}O{x}}{\good[E][#1][#2]}
\NewDocumentCommand{\goods}{O{}O{x}}{\good[S][#1][#2]}
\NewDocumentCommand{\goodeinc}{O{x}}{\goode[\strinc][#1]}
\NewDocumentCommand{\goodsinc}{O{x}}{\goods[\strinc][#1]}
\NewDocumentCommand{\goodedec}{O{x}}{\goode[\strdec][#1]}
\NewDocumentCommand{\goodsdec}{O{x}}{\goods[\strdec][#1]}
\NewDocumentCommand{\goodezero}{O{x}}{\goode[\strzero][#1]}
\NewDocumentCommand{\goodszero}{O{x}}{\goods[\strzero][#1]}

\mathcommand{\ko}{\text{ko}}
\mathcommand{\ok}{\text{ok}}
\mathcommand{\noop}{\text{noop}}
\mathcommand{\koe}{\ko_E}
\mathcommand{\oke}{\ok_E}
\mathcommand{\kos}{\ko_S}
\mathcommand{\oks}{\ok_S}

\newlist{enuE}{enumerate}{1}
\setlist[enuE]{label=\textbf{[E\arabic*]}}
\newlist{enuS}{enumerate}{1}
\setlist[enuS]{label=\textbf{[S\arabic*]}}

\NewDocumentCommand{\ce}{O{i}}{\ensuremath{e_{#1}}\xspace}
\NewDocumentCommand{\cs}{O{i}}{\ensuremath{s_{#1}}\xspace}
\NewDocumentCommand{\cnt}{O{i}}{\ensuremath{c_{#1}}\xspace}
\mathcommand{\cntzero}{\cnt[0]}
\mathcommand{\cntone}{\cnt[1]}

\mathcommand{\eproc}{\bullet}

\NewDocumentCommand{\inc}{O{i}}{\ensuremath{\text{inc}_{#1}}\xspace}
\NewDocumentCommand{\dec}{O{i}}{\ensuremath{\text{dec}_{#1}}\xspace}

\mathcommand{\formkoe}{\Phi_{\koe}}
\mathcommand{\formkos}{\Phi_{\kos}}
\NewDocumentCommand{\formbadseq}{O{x}}{\ensuremath{\Psi_{\text{bad sequence}}(#1)}\xspace}
\NewDocumentCommand{\formbadtarget}{O{x}}{\ensuremath{\Psi_{\text{bad target}}(#1)}\xspace}
\NewDocumentCommand{\formbadsource}{O{x}}{\ensuremath{\Psi_{\text{bad source}}(#1)}\xspace}
\NewDocumentCommand{\formbadupkeep}{O{x}}{\ensuremath{\Psi_{\text{bad upkeep}}(#1)}\xspace}
\NewDocumentCommand{\formbadzerotest}{O{x}}{\ensuremath{\Psi_{\text{bad zero test}}(#1)}\xspace}

\NewDocumentCommand{\isstate}{O{x}}{\ensuremath{\States(#1)}\xspace}
\NewDocumentCommand{\istrans}{O{x}}{\ensuremath{\Trans(#1)}\xspace}
\NewDocumentCommand{\istransinc}{O{x}}{\ensuremath{\Transinc(#1)}\xspace}
\NewDocumentCommand{\istransdec}{O{x}}{\ensuremath{\Transdec(#1)}\xspace}
\NewDocumentCommand{\istranszero}{O{x}}{\ensuremath{\Transzero(#1)}\xspace}
\NewDocumentCommand{\isupkeep}{O{x}}{\ensuremath{\mathcal U(#1)}\xspace}
\NewDocumentCommand{\issys}{O{x}}{\ensuremath{\sigS(#1)}\xspace}
\NewDocumentCommand{\isenv}{O{x}}{\ensuremath{\sigE(#1)}\xspace}

\mathcommand{\formprefixE}{\Phi_{\text{bad prefix E}}}
\mathcommand{\formprefixS}{\Phi_{\text{bad prefix S}}}
\mathcommand{\formblockE}{\Phi_{\text{E blocks}}}
\mathcommand{\formblockS}{\Phi_{\text{S blocks}}}
\mathcommand{\formplayafterkoE}{\Phi_{\text{E plays after ko}}}
\mathcommand{\formplayafterkoS}{\Phi_{\text{S plays after ko}}}

\mathcommand{\first}{\texttt{first}}
\mathcommand{\second}{\texttt{second}}
\mathcommand{\last}{\texttt{last}}

\mathcommand{\N}{\mathbb{N}}
\mathcommand{\Zero}{\mathbb O}
\mathcommand{\im}{\operatorname{Im}}
\newcommand{\EF}{Ehrenfeucht-Fra\"iss\'e\xspace}

\mathcommand{\logic}{\mathcal L}
\mathcommand{\logicbis}{\logic'}

\mathcommand{\thr}{t}

\mathcommand{\elem}{a}
\mathcommand{\elembis}{b}
\mathcommand{\elemter}{c}

\mathcommand{\var}{x}
\mathcommand{\varbis}{y}

\mathcommand{\classe}{\mathcal C}
\mathcommand{\prop}{\mathcal P}

\mathcommand{\structdom}{A}
\mathcommand{\struct}{\mathcal A}
\mathcommand{\structbisdom}{A'}
\mathcommand{\structbis}{\mathcal A'}

\mathcommand{\type}{\tau}
\mathcommand{\typebis}{\type'}
\mathcommand{\typeter}{\type''}
\NewDocumentCommand{\typeghost}{O{\play}}{\ensuremath{\type^{\ghost[]}_{#1}}\xspace}
\mathcommand{\nexttypeghost}{\typeghost[\nextplay]}

\mathcommand{\formule}{\varphi}
\mathcommand{\formulebis}{\psi}
\mathcommand{\formuleter}{\theta}
\mathcommand{\vocab}{\Sigma}

\mathcommand{\graphe}{\mathcal\graphedom}
\mathcommand{\graphebis}{\mathcal\graphebisdom}
\mathcommand{\graphedom}{G}
\mathcommand{\graphebisdom}{\graphedom'}
\mathcommand{\graphedombis}{H}
\mathcommand{\edgerel}{E}

\mathcommand{\motvide}{\varepsilon}

\NewDocumentCommand{\initsegment}{O{n}}{\ensuremath{[1,#1]}\xspace}
\mathcommand{\oldsysproc}{\initsegment[\nS+1]}
\mathcommand{\newsysproc}{\initsegment[\nS]}
\mathcommand{\envproc}{\initsegment[\nE]}

\mathcommand{\oldProcS}{\ProcS^{\nS+1}}
\mathcommand{\newProcS}{\ProcS^{\nS}}

\mathcommand{\oldTokS}{\TokS^{\nS+1}}
\mathcommand{\newTokS}{\TokS^{\nS}}

\NewDocumentCommand{\ps}{O{}O{}}{\ensuremath{s_{#1}^{#2}}\xspace}
\NewDocumentCommand{\oldps}{O{}}{\ps[#1][\nS+1]}
\NewDocumentCommand{\newps}{O{}}{\ps[#1][\nS]}
\NewDocumentCommand{\pe}{O{}}{\ensuremath{e_{#1}}\xspace}

\NewDocumentCommand{\strat}{O{}}{\ensuremath{\mathcal S_{#1}}\xspace}
\mathcommand{\stratnEplus}{\strat[+]}
\mathcommand{\stratnE}{\strat[]}
\mathcommand{\oldstrat}{\strat[+]}
\mathcommand{\newstrat}{\strat[]}

\mathcommand{\gameplus}{\game_{+}}

\mathcommand{\largetype}{\mathit{lt}}

\NewDocumentCommand{\map}{O{\play}}{\ensuremath{\sigma_{#1}}\xspace}
\mathcommand{\nextmap}{\map[\nextplay]}

\NewDocumentCommand{\oldplay}{O{\play}}{\ensuremath{\nu({#1})}\xspace}
\mathcommand{\nextoldplay}{\oldplay[\nextplay]}

\NewDocumentCommand{\ghost}{O{\play}}{\ensuremath{\gamma_{#1}}\xspace}
\mathcommand{\nextghost}{\ghost[\nextplay]}      

\NewDocumentCommand{\actarray}{O{}}{\ensuremath{\texttt{act}_{#1}}\xspace}
\mathcommand{\prevactarray}{\actarray[r]}
\mathcommand{\nextactarray}{\actarray[r+1]}
\NewDocumentCommand{\act}{O{}O{\ps}}{\ensuremath{\actarray[#1][#2]}\xspace}
\NewDocumentCommand{\prevact}{O{\ps}}{\ensuremath{\act[r][#1]}\xspace}
\NewDocumentCommand{\nextact}{O{\ps}}{\ensuremath{\act[r+1][#1]}\xspace}

\NewDocumentCommand{\invstp}{O{\play}}{\ensuremath{(\texttt{I}_\text{S}[#1])}\xspace}
\NewDocumentCommand{\invetp}{O{\play}}{\ensuremath{(\texttt{I}_\text{E}[#1])}\xspace}
\NewDocumentCommand{\invnu}{O{\play}}{\ensuremath{(\texttt{I}_{+} [#1])}\xspace}
\NewDocumentCommand{\invbound}{O{\play}}{\ensuremath{(\texttt{I}_\text{F} [#1])}\xspace}
\NewDocumentCommand{\invghosttp}{O{\play}}{\ensuremath{(\texttt{I}_{\typeghost[]} [#1])}\xspace}
\NewDocumentCommand{\invmax}{O{\play}}{\ensuremath{(\texttt{I}_{\text{max}} [#1])}\xspace}
\NewDocumentCommand{\invbig}{O{\play}}{\ensuremath{(\texttt{I}_{\text{big}} [#1])}\xspace}

\mathcommand{\nextinvstp}{\invstp[\nextplay]}
\mathcommand{\nextinvetp}{\invetp[\nextplay]}
\mathcommand{\nextinvnu}{\invnu[\nextplay]}
\mathcommand{\nextinvbound}{\invbound[\nextplay]}
\mathcommand{\nextinvghosttp}{\invghosttp[\nextplay]}
\mathcommand{\nextinvmax}{\invmax[\nextplay]}
\mathcommand{\nextinvbig}{\invbig[\nextplay]}

\NewDocumentCommand{\dw}{O{}O{}}{\ensuremath{\dword_{#1}^{#2}}\xspace}
\mathcommand{\prevolddw}{\dw[r][\nS+1]}
\mathcommand{\nextolddw}{\dw[r+1][\nS+1]}
\mathcommand{\prevnewdw}{\dw[r][\nS]}
\mathcommand{\nextnewdw}{\dw[r+1][\nS]}

\mathcommand{\seqmoves}{\widehat{m}}
\mathcommand{\seqmovesbis}{\seqmoves'}

\NewDocumentCommand{\updatetp}{O{\dword}O{\ps}O{\type}}{\ensuremath{#1\big[#2\Rightarrow #3\big]\xspace}}

\NewDocumentCommand{\msotp}{O{p}O{\dw}O{k}}{\ensuremath{\text{mso-tp}^{#3}(#1\,|\,#2)}\xspace}
\NewDocumentCommand{\prevoldmsotp}{O{\prevmap(s)}O{k}}{\msotp[#1][\prevolddw][#2]}
\NewDocumentCommand{\nextoldmsotp}{O{\nextmap(s)}O{k}}{\msotp[#1][\nextolddw][#2]}
\NewDocumentCommand{\prevnewmsotp}{O{s}O{k}}{\msotp[#1][\prevnewdw][#2]}
\NewDocumentCommand{\nextnewmsotp}{O{s}O{k}}{\msotp[#1][\prevnewdw][#2]}

\NewDocumentCommand{\preffotp}{O{p}O{\dw}O{k}}{\ensuremath{\langle #1\,|\,#2\rangle^{#3}_{\prefFO}}\xspace}
\NewDocumentCommand{\prevoldpreffotp}{O{\prevmap(s)}O{k}}{\preffotp[#1][\prevolddw][#2]}
\NewDocumentCommand{\nextoldpreffotp}{O{\nextmap(s)}O{k}}{\preffotp[#1][\nextolddw][#2]}
\NewDocumentCommand{\prevnewpreffotp}{O{s}O{k}}{\preffotp[#1][\prevnewdw][#2]}
\NewDocumentCommand{\nextnewpreffotp}{O{s}O{k}}{\preffotp[#1][\nextnewdw][#2]}

\NewDocumentCommand{\height}{O{\type}}{\ensuremath{h(#1)}\xspace}
\mathcommand{\heightmax}{h_{\mathrm{max}}}

\NewDocumentCommand{\boundcc}{O{\height}}{\ensuremath{F_\text{CC}(#1)}\xspace}
\NewDocumentCommand{\boundtp}{O{\type}}{\ensuremath{F(#1)}\xspace}  
\NewDocumentCommand{\bound}{O{\height}}{\ensuremath{F(#1)}\xspace}  

\NewDocumentCommand{\CC}{O{\type}}{\ensuremath{\text{CC}(#1)}\xspace}
\mathcommand{\CCnp}{\mathcal{C}}
\mathcommand{\prevghostcc}{\CC[\oldmsotp[\prevghost]]}
\mathcommand{\nextghostcc}{\CC[\oldmsotp[\nextghost]]}

\NewDocumentCommand{\tpord}{O{k}}{\ensuremath{\prec_{#1}}\xspace}

\NewDocumentCommand{\auto}{O{k}}{\ensuremath{\mathfrak A_{#1}}\xspace}
\NewDocumentCommand{\preordtp}{O{k}}{\ensuremath{\rightharpoonup_{#1}}\xspace}
\NewDocumentCommand{\equivtp}{O{k}}{\ensuremath{\rightleftharpoons_{#1}}\xspace}

\mathcommand{\false}{\texttt{false}}
\mathcommand{\true}{\texttt{true}}

\mathcommand{\seqE}{\dw[E]}
\mathcommand{\seqS}{\dw[S]}
\mathcommand{\longseq}{\widehat{\dw}}

\mathcommand{\resetact}{\texttt{reset\_act()}}
\mathcommand{\restr}{\upharpoonright}

\mathcommand{\perm}{\xi}
\mathcommand{\permbis}{\perm'}

\usepackage{scalerel}
\mathcommand{\letterE}{a_E}
\mathcommand{\letterS}{a_S}
\mathcommand{\exformule}{\formule_{{\scriptscriptstyle 1 a_{\scaleto{E}{2pt}}\leftrightarrow1a_{\scaleto{S}{2pt}}}}}
\mathcommand{\open}{\texttt{open}}
\mathcommand{\close}{\texttt{close}}

\begin{document}

\title{Synthesis for prefix first-order logic on data words
\thanks{Supported by the ERC Consolidator grant D-SynMA (No. 772459).}}

\ifcsname fullversion\endcsname

\author{Julien Grange\footnote{Univ Paris Est Creteil, LACL, F-94010 Creteil, France}\and
Mathieu Lehaut\footnote{University of Gothenburg, Sweden}}
\date{}

\else

\author{Julien Grange\inst{1}\orcidID{0009-0005-0470-1781} \and
Mathieu Lehaut\inst{2}\orcidID{0000-0002-6205-0682}}
\authorrunning{J. Grange and M. Lehaut}
\institute{Univ Paris Est Creteil, LACL, F-94010 Creteil, France 
\email{julien.grange@lacl.fr}
\and
University of Gothenburg, Sweden
\email{lehaut@chalmers.se}}

\fi

\maketitle

\begin{abstract}
We study the reactive synthesis problem for distributed systems with an unbounded number of participants interacting with an uncontrollable environment.
Executions of those systems are modeled by data words, and specifications are given as first-order logic formulas from a fragment we call prefix first-order logic that implements a limited kind of order.
We show that this logic has nice properties that enable us to prove decidability of the synthesis problem.
\end{abstract}

\section{Introduction}\label{sec:intro}
Distributed algorithms have been increasingly more common in recent years, and can be found in a wide range of domains such as distributed computing, swarm robotics, multi-agent systems, and communication protocols, among others.
Those algorithms are often more complex than single-process algorithms due to the interplay between the different processes involved in the computations.
Another complication arises in the fact that some algorithms must be designed for distributed systems where the number of participants is not known in advance, which is often the case in applications where agents can come and leave at a moment's notice as is the case, for instance, in ad-hoc networks.
Those properties make it hard for programmers to design such algorithms without any mistake, thus justifying the development of formal methods for their verification.

In this paper, we focus on the \emph{reactive synthesis} problem.
This problem involves systems that interact with an uncontrollable environment, with the system outputting some values depending on the inputs that are given by the environment.
Given a specification stating what are the allowed behaviors of the whole system, the goal is to automatically build a program that would satisfy the specification.
This problem dates all the way back to Church \cite{sisl1957-Chu}, whose original statement focused on sequential systems, and was solved in this context by Büchi and Landweber \cite{TAMS138-BL}.
They reformulated this problem as a two-player synthesis game between the System and an adversarial Environment alternatively choosing actions from a finite alphabet.
The goal of System is for the resulting sequence of actions to satisfy the specification, while Environment wants to falsify it.
A winning strategy for System, if it exists, can then be seen as a program ensuring that the specification is always met.

At the cost of having one copy of the alphabet for each participant, one can adapt this setting to distributed systems where the number of participants is fixed. However, a finite alphabet is too restrictive to handle systems with an unbounded number of participants such as those we described earlier.
Indeed, with a finite amount of letters in the alphabet, one cannot distinguish every possible participant when there can be any number of those, potentially more than the number of letters.
At most, one can deal by grouping participants together in a finite number of classes and consider all that belong to the same class to be equivalent.
This however restrict the behaviors of the system and what one could specify over the system.
It is therefore worth extending this problem to infinite action alphabets.
To that end, we turn to \emph{data words}, as introduced by Bojanczyk et al. \cite{DBLP:conf/lics/BojanczykMSSD06}.
A data word is a sequence of pairs consisting of an action from a finite alphabet and a datum from an infinite alphabet. In our context the datum represents the identity of the process doing the action, meaning that a data word is seen as an execution of the system describing sequentially which actions have been taken by each process.
In the corresponding synthesis game, the two players alternatively choose both an action and a process, and the resulting play is a data word.

The last ingredient needed to properly define the synthesis problem is the choice of a formalism in which specifications are written.
Unfortunately, there is no strong candidate for a ``standard'' way of representing sets of data words, in contrast to finite automata in the case of simple words.
Many formalisms have been proposed so far, but all of them lack either good closure properties (union, complementation, etc.), have bad complexity for some basic decision procedures (membership, ...), are lacking in expressivity power or do not have a good equivalent automata $\Leftrightarrow$ logic characterization.
Let us cite nonetheless register automata who were considered first by Kaminsky and Francez \cite{kaminski1994finite} but have seen different extensions over time, pebble automata by Neven et al. \cite{neven2004finite} and data automata \cite{DBLP:conf/lics/BojanczykMSSD06}.
On the logical side, several formalisms have been proposed, such as a variant of first-order logic \cite{DBLP:conf/lics/BojanczykMSSD06}, Freeze LTL~\cite{demri2009ltl} and the logic of repeating values \cite{demri2012temporal}.
We refer the reader to Ahmet Kara's dissertation for a more comprehensive survey \cite{kara2016logics}.
Most of the previous works study the membership problem (in the case of automata) and the satisfiability problem (for logics), which are useful for model-checking applications but not enough in the synthesis context, as they lack the adversarial environment factor that is central to the problem.
A few attempts have been made in this direction, notably for the logic of repeating values by Figueira and Praveen \cite{FigueiraP18} and for register automata by Exibard et al. \cite{exibard2022generic}.

We follow previous work \cite{DBLP:conf/fossacs/BerardBLS20,grange2023first} and focus on synthesis for first-order logic extended to data words.
In this extension, we add a predicate $\sim$ to the logic such that $x \sim y$ is true when two positions $x$ and $y$ of a data word share the same data value, i.e. when both actions have been made by the same process.
Moreover, we partition the set of processes into System processes and Environment processes, and restrict each player to their own processes.
The reason for this is two-fold: first, when processes are shared, the synthesis problem has been shown to be undecidable even for the simplest logic possible $\FO^2[\sim]$, where only two variable names are allowed.
Second, inputs and outputs are usually physically located in different components of the system, such as sensors being disjoint from motors in a drone; it thus makes sense to see them as different processes.

As always, there is a trade-off between expressivity of the logic and decidability of its synthesis problem.
The well-known LTL undecidability result from Pnueli and Rosner \cite{pneuli1990distributed} occurs because the logic allows specifying properties that the system is too weak to satisfy.
We must therefore find a balance between making the logic expressive enough to be useful, while limiting its power so that it cannot specify properties the system is not expected to be able to satisfy in the first place.
Our previous results (\cite[Theorems 3 and 4]{grange2023first}) have shown that adding any kind of order on the positions of the data word, such as either the immediate successor predicate or the happens-before predicate, makes the synthesis problem undecidable.
While those two predicates are fine from a centralized point of view that sees the whole system as a sequential machine, they are not that well suited for real-life systems.
Indeed, it is ambitious to expect each process to know everything that happened on other processes in the exact order those actions happened; this would require every process to be informed instantly after every action happening in the system.
What is more reasonable is simply to expect a process to know the order of occurrence of its own actions only, without knowing whether those actions happened before or after actions made by other processes.

This leads us to introduce a new operator $\lesssim$, for which $x \lesssim y$ if $x \sim y$ and $x$ occurs before $y$ in the data word.
We call $\FO[\lesssim]$ the extension of \FO that includes only this new predicate.
It is strictly more expressive than $\FO[\sim]$, as the $\sim$ predicate can easily be simulated by $\lesssim$.
Similar to $\FO[\sim]$, we can study separately the class of each process.
Whereas $\FO[\sim]$ could only count how many times each action happened (up to some threshold), we can now express anything that $\FO[<]$ can express on (simple) words.
Unfortunately, the decidability of the synthesis problem for $\FO[\lesssim]$ remains open.
In this paper, we show a positive result for a restriction of this logic that we call \emph{prefix first-order logic}, denoted by \prefFOdw.
In this restriction, the first variable quantified for each process is called a bounding variable, and every subsequently defined variable belonging to the same process must occur before the bounding variable.
In other words, the first variable for each class pins down a finite prefix of the class and throws away the rest of the class; the rest of the variables can only talk about positions that fall inside this prefix.
This restriction allows us to obtain good properties that we leverage to show decidability of the synthesis problem.

This paper is organized as follows. 
We first define prefix first-order logic \prefFOdw in Section~\ref{sec:order}, and the synthesis problem and its equivalent games in Section~\ref{sec:synthesis}.
We then show the synthesis problem for \prefFOdw to be decidable of in Section~\ref{sec:decidability}, before concluding in Section~\ref{sec:conclu}.
\ifcsname fullversion\endcsname
Omitted proofs can be found in the appendix.
\fi

\section{Prefix first-order logic}\label{sec:order}
\subsection{Data words and preliminaries}
Fix two disjoint alphabets $\vocabS$ and $\vocabE$, which are respectively the \Sys and \Envi \emph{actions}.
Let $\vocab = \vocabS \uplus \vocabE$ denote their union. In the following, we write \motvide for the empty word, and $u\cdot v$ for the concatenation of the words $u$ and $v$.
Let $\ProcS$ and $\ProcE$ be two disjoint sets of \Sys and \Envi \emph{processes}, respectively.

A \emph{data word} is a (finite or infinite) sequence $\dw = (a_0,p_0) (a_1,p_1) \dots$ of pairs $(a_i,p_i) \in (\vocabS \times \ProcS) \cup (\vocabE \times \ProcE)$.
A pair $(a,p)$ indicates that action $a$ has been taken by process $p$.
The \emph{class} of a process $p$ is the word $w_p = a_{k_0} a_{k_1} \dots \in \vocab^\star$ where $(k_i)_i$ is exactly the sequence of positions in $\dw$ where actions are made by $p$.
We see data words as logical structures (and will conveniently identify a data word with its associated structure) over the vocabulary consisting of two binary predicates $\sim$, $<$, and two unary predicates $\ProcS$, $\ProcE$ as well as one additional unary predicate for each letter of \vocab. The universe of a data word has one element for each process (called \emph{process elements} -- these are needed in order to quantify over processes that have not played any action, and will drastically increase the expressive power of \prefFOdw), and one element for each position in the data word. Predicate \ProcS (resp. \ProcE) is interpreted as the set of all \Sys (resp. \Envi) process elements. For  $a\in\vocab$, the predicate $a$ holds on every position which correspond to action $a$. Predicate $<$ is interpreted as the linear order on the set of positions corresponding to their order in the data word (and is thus not defined on process elements), and $\sim$ is interpreted as an equivalence relation which has one equivalence class for every process, encompassing both its process element and all positions of its class.
It will be convenient to use $\Proc(x)$ as an alias for ``$\ProcS(x)\lor\ProcE(x)$'' and $x\lesssim y$ as an alias for ``$x\sim y\land x<y$''.

Let $\DW[\vocab][\Proc]$ denote the set of all data words over actions $\vocab$ and processes $\Proc$.
We write $\dw[][][i \dots j]$ for the factor of $\dw$ occurring between positions $i$ and $j$ (both included), and $\dw[][][i \dots]$ for the suffix starting at position $i$; we extend both notations to regular words as well.

\subsection{Prefix first-order logic on data words}
\label{sec:def_preffo}

We define \emph{prefix first-order logic on data words} \prefFOdw by induction on its formulas. We write $\formule(\procVar;\boundVar;\prefVar)$ to mean that the free variables of \formule belong to the pairwise disjoint union \Var of the three sets \procVar (the \emph{process variables}), \boundVar (the \emph{bounding variables}) and \prefVar (the \emph{prefix variables}). 
\ifcsname fullversion\endcsname
\begin{equation*}
  \begin{aligned}
    &\formule(\procVar;\ \boundVar;\prefVar)::=\\
    & \quad x=y &(x,y\in\Var)\\
    &|\ \ProcS(x)\quad |\ \ProcE(x)&(x\in\procVar)\\
    &|\ a(x) &(x\in\boundVar\cup\prefVar, a\in\vocab)\\
    &|\ x\lesssim y&(x,y\in\prefVar)\\
    &|\ x\sim y&((x,y)\in\procVar\times\Var)\\
    &|\ \formule(\procVar;\boundVar;\prefVar) \land \formule(\procVar;\boundVar;\prefVar)\\
    &|\ \neg \formule(\procVar;\boundVar;\prefVar)\\
    &|\ \exists x,\Proc(x)\ \land\ \formule(\procVar\cup\{x\};\boundVar;\prefVar)&(x\notin\Var)\\
    &|\ \exists x, \neg\Proc(x)\ \land\ \Big(\bigwedge\limits_{y\in\boundVar}x\not\sim y\Big)\land \formule(\procVar;\boundVar\cup\{x\};\prefVar)&(x\notin\Var)\\
    &|\ \exists x,\ x\lesssim y\land \formule(\procVar;\boundVar;\prefVar\cup\{x\})&(x\notin\Var, y\in\boundVar)
  \end{aligned}\\
\end{equation*}
\else
\begin{equation*}
  \begin{aligned}
    &\formule(\procVar;\ \boundVar;\prefVar)::=\\
    & \quad x=y \hspace{.764\textwidth}(x,y\in\Var)\\
    &|\ \ProcS(x)\quad |\ \ProcE(x)\quad \hspace{.609\textwidth}(x\in\procVar)\\
    &|\ a(x) \hspace{.594\textwidth}(x\in\boundVar\cup\prefVar, a\in\vocab)\\
    &|\ x\lesssim y\hspace{.735\textwidth}(x,y\in\prefVar)\\
    &|\ x\sim y\hspace{.646\textwidth}((x,y)\in\procVar\times\Var)\\
    &|\ \formule(\procVar;\boundVar;\prefVar) \land \formule(\procVar;\boundVar;\prefVar)\\
    &|\ \neg \formule(\procVar;\boundVar;\prefVar)\\
    &|\ \exists x,\Proc(x)\ \land\ \formule(\procVar\cup\{x\};\boundVar;\prefVar) \hspace{.371\textwidth}(x\notin\Var)\\
    &|\ \exists x, \neg\Proc(x)\ \land\ \Big(\bigwedge\limits_{y\in\boundVar}x\not\sim y\Big)\land \formule(\procVar;\boundVar\cup\{x\};\prefVar)\hspace{.124\textwidth} (x\notin\Var)\\
    &|\ \exists x,\ x\lesssim y\land \formule(\procVar;\boundVar;\prefVar\cup\{x\}) \hspace{.24\textwidth}(x\notin\Var, y\in\boundVar)
  \end{aligned}\\
\end{equation*}
\fi
with $\Var = \procVar \uplus \boundVar \uplus \prefVar$. The intuition is as follows.

We allow one to quantify over the process elements with process variables. Bounding and prefix variables are used to quantify over the elements of the process classes; when quantifying (existentially or universally) over an element of a process class (that is, an actual position of the data word), one must first use a bounding variable. From then on, only prefix variables can be used on this process class, which can only quantify earlier positions in the class (i.e. positions which are $\lesssim$ to the bounding position). Note that one can still quantify over other classes, using new bounding variables.

The semantics is defined as usual. As always, the \emph{quantifier depth} of a formula is the maximal number of nested quantifiers (without regard to whether they quantify process, bounding or prefix variables).

By construction, \prefFOdw is a fragment of \FOeqord (first-order logic with $\sim$ and $<$). Example~\ref{ex:main_ex} below illustrates that \prefFOdw encompasses \FOeq.

\begin{example}
  \label{ex:main_ex}
  Let us fix the alphabets $\vocabE:=\{\letterE\}$ and $\vocabS:=\{\letterS\}$. There exists an \prefFOdw formula \exformule of quantifier depth $3$ stating that there exists a process with exactly one \letterE if and only if there exists a process with exactly one \letterS. Indeed, the existence of a process with exactly one \letterE (and similarly for \letterS) can be stated as
  \begin{equation*}
    \begin{aligned}
      \exists x,\ \Proc(x)\quad &\land\ \ \quad \big(\exists y,\ \neg\Proc(y)\ \land\ y\sim x\ \land\ \letterE(y)\big) \\
      &\land\quad\neg \big(\exists y,\ \neg\Proc(y)\ \land\ y\sim x\ \land\ \letterE(y)\ \land\ \exists z,\ z\lesssim y\land\letterE(z) \big)\,.
    \end{aligned}
  \end{equation*}
    Here, $x$ is a process variable, both occurrences of $y$ are bounding variables and $z$ is a prefix variable.
\end{example}

If a (finite or infinite) word is seen as a data word with exactly one data class, then \prefFOdw can in particular be seen a logic on words: it is equivalent to the restriction of first-order logic on words where
\begin{itemize}
\item the formula must start with a universal or existential quantification on the variable $\overline x$,
\item after that, each new quantification must be of the form $\exists x<\overline x$ or $\forall x<\overline x$.
\end{itemize}
We will refer to this logic on words as \prefFO.

\begin{example}
  In order to get a better understanding of the expressive power of \prefFO and its limitation, let us consider a finite alphabet containing, among others, the two symbols \open and \close.
\\
  The property stating that every occurrence of \close must be preceded by an occurrence of \open can be formulated as follows in \prefFO:
  \[\forall \bar x,\ \close(\bar x)\quad \to\quad \exists y,\ y<\bar x\ \land\ \open(y)\,.\]
  In contrast, one cannot state in \prefFO that each occurrence of \open is followed by an occurrence of \close.
  Indeed, for a fixed quantifier depth, the two words $w = \open \cdot \close \cdot \open \cdot ... \cdot \close$ and $w' = w \cdot \open$ cannot be distinguished if $w$ is long enough.
\end{example}

For any word $w$, we let $\tpword[w]$ denote its \prefFO-type of depth $k$, i.e. the set of all sentences of \prefFO with quantifier depth at most $k$ satisfied in $w$.
Having fixed an alphabet, we denote by \Types the set of \prefFO-types of depth $k$ on words.

\begin{toappendix}

  When we defined \prefFOdw in Section~\ref{sec:def_preffo}, we did not considered constant symbols, as the game from Section~\ref{sec:standard_game} does not deal with constants -- note that considering constants in these games would be cumbersome without any increase in expressivity, as one can already implement them by adding new letters to the alphabets \vocabS and \vocabE and ensure their uniqueness in the formula \formule.

  It will however be convenient to allow the use of constant symbols in \prefFOdw, in order to make the proof of Proposition~\ref{prop:EF} smoother. Let us thus extend the definition of \prefFOdw in the following way.

  We add to the vocabulary of data words a finite set of constant symbols \Const, which we partition into \procConst (the \emph{process constants}), \boundConst (the \emph{bounding constants}) and \prefConst (the \emph{prefix constants}). We restrict our attention to data words in which the interpretation of these constant symbols is such that
  \begin{itemize}
  \item each process constant is interpreted as a process,
  \item each bounding and prefix are interpreted as positions of the data word,
  \item no two bounding constants are in the same data class, and
  \item each prefix constant is in the same data class as some bounding constant, which appear at a later position (with respect to $\leq$).
  \end{itemize}
  
  We define as before the formulas of \prefFOdw by induction, taking now into account the constant symbols. Again, $\formule(\procVar;\boundVar;\prefVar)$ means that the free variables of \formule belong to the pairwise disjoint union \Var of the three sets \procVar (the \emph{process variables}), \boundVar (the \emph{bounding variables}) and \prefVar (the \emph{prefix variables}). The additions to the original definition are highlighted.
  \ifcsname fullversion\endcsname
  \begin{equation*}
    \begin{aligned}
      &\quad x=y&(x,y\in\Var\mybox{\cup\Const})\\
      &|\ \ProcS(x)\quad |\ \ProcE(x)&(x\in\procVar\mybox{\cup\procConst})\\
      &|\  a(x)&(x\in\boundVar\cup\prefVar\mybox{\cup\boundConst\cup\prefConst}, a\in\vocab)\\
      &|\  x\lesssim y&(x,y\in\prefVar\mybox{\cup\prefConst})\\
      &|\  x\sim y& ((x,y)\in(\procVar\cup\procConst)\times(\Var\mybox{\cup\Const}))\\
      &|\  \formule(\procVar;\boundVar;\prefVar) \land \formule(\procVar;\boundVar;\prefVar)\\
      &|\  \neg \formule(\procVar;\boundVar;\prefVar)\\
      &|\  \exists x,\Proc(x)\ \land\ \formule(\procVar\cup\{x\};\boundVar;\prefVar)& (x\notin\Var)\\
      &|\  \exists x, \neg\Proc(x)\ \land\ \Big(\bigwedge\limits_{y\in\boundVar\mybox{\cup\boundConst}}x\not\sim y\Big)\land \formule(\procVar;\boundVar\cup\{x\};\prefVar)&(x\notin\Var)\\
      &|\  \exists x,\ x\lesssim y\land \formule(\procVar;\boundVar;\prefVar\cup\{x\})&(x\notin\Var, y\in\boundVar\mybox{\cup\boundConst})
    \end{aligned}
  \end{equation*}
  \else
  \begin{equation*}
    \begin{aligned}
      &\quad x=y \hspace{.71\textwidth}(x,y\in\Var\mybox{\cup\Const})\\
      &|\ \ProcS(x)\quad |\ \ProcE(x)\hspace{.538\textwidth}(x\in\procVar\mybox{\cup\procConst})\\
      &|\  a(x) \hspace{.405\textwidth}(x\in\boundVar\cup\prefVar\mybox{\cup\boundConst\cup\prefConst}, a\in\vocab)\\
      &|\  x\lesssim y \hspace{.64\textwidth} (x,y\in\prefVar\mybox{\cup\prefConst})\\
      &|\  x\sim y\hspace{.4465\textwidth} ((x,y)\in(\procVar\cup\procConst)\times(\Var\mybox{\cup\Const}))\\
      &|\  \formule(\procVar;\boundVar;\prefVar) \land \formule(\procVar;\boundVar;\prefVar)\\
      &|\  \neg \formule(\procVar;\boundVar;\prefVar)\\
      &|\  \exists x,\Proc(x)\ \land\ \formule(\procVar\cup\{x\};\boundVar;\prefVar) \hspace{.3765\textwidth} (x\notin\Var)\\
      &|\  \exists x, \neg\Proc(x)\ \land\ \Big(\bigwedge\limits_{y\in\boundVar\mybox{\cup\boundConst}}x\not\sim y\Big)\land \formule(\procVar;\boundVar\cup\{x\};\prefVar) \hspace{.0304\textwidth} (x\notin\Var)\\
      &|\  \exists x,\ x\lesssim y\land \formule(\procVar;\boundVar;\prefVar\cup\{x\})  \hspace{.147\textwidth} (x\notin\Var, y\in\boundVar\mybox{\cup\boundConst})
    \end{aligned}
  \end{equation*}
  
  \fi

  As intended, the logic \prefFOdw we defined in Section~\ref{sec:def_preffo} is just the restriction of the above definition when the set \Const of constant symbols is empty.

  We write $\dword\prefFOdweq\dwordbis$ when \dword and \dwordbis agree on every \prefFOdw-sentence of quantifier depth at most $k$.

  As per usual, when introducing a new logic, it is useful to find a characterization via an \EF game. We introduce the suiting \EF game below, and show its equivalence to \prefFOdw in Proposition~\ref{prop:EF}.

\begin{definition}
  Let $k$ be an integer, and $\dword, \dwordbis$ be two (finite or infinite) data words on the same alphabet with constant set \Const. The \emph{$k$-round prefix \EF game} on data words is played by two players: the Spoiler, who tries to highlight the differences between \dword and \dwordbis, and the Duplicator, whose goal is to show that \dword and \dwordbis look alike. There are three kinds of pebbles available to place on each data word, to which both players have access: \emph{process pebbles} \procpebble (in \dword) and \procpebblebis (in \dwordbis), \emph{bounding pebbles} \boundpebble (in \dword) and \boundpebblebis (in \dwordbis) and \emph{prefix pebbles} \prefpebble (in \dword) and \prefpebblebis (in \dwordbis), for $i\in\{1,\dots,k\}$. These pebbles will be placed on elements of \dword and \dwordbis by the players. With a slight abuse of notation, we will sometimes identify a pebble with the element on which it has been placed.

  In round $i$, the Spoiler starts by choosing one of the two data words. Let us assume first that they choose to play in \dword. Then the Spoiler has three possibilities: they can make a \emph{process move}, a \emph{bounding move} or a \emph{prefix move} by respectively placing \procpebble, \boundpebble or \prefpebble on some element of \dword, with the following restriction.
  \begin{enumerate}
  \item In a process move, \procpebble must be placed on a process element of \dword.
  \item In a bounding move, pebble \boundpebble must be placed on a position of \dword whose class does not already possess a bounding pebble, nor a bounding constant.
  \item In a prefix move, the Spoiler must place \prefpebble on an position $e$ of \dword such that either some bounding pebble \boundpebble[j] (for $j<i$) or some bounding constant is in the same class as $e$ and larger than $e$ with respect to $\leq$. 
  \end{enumerate}
  Then the Duplicator responds by placing the corresponding pebble in \dwordbis (i.e. \procpebblebis if the Spoiler placed \procpebble, \boundpebblebis if they placed \boundpebble, or \prefpebblebis if they placed \prefpebble) with the same constraints.
  Conversely, if the Spoiler decided to play in \dwordbis, then the Duplicator would respond by placing either \procpebble, \boundpebble or \prefpebble (depending on whether the Spoiler played \procpebblebis, \boundpebblebis or \prefpebblebis) in \dword.

  At the end of round $i$, let us assume that pebbles $\pebble[1],\dots,\pebble$ have been played on \dword (where \pebble[j] is one of \procpebble[j], \boundpebble[j] or \prefpebble[j], depending on whether round $j$ saw a process move, a bounding move or a prefix move), and $\pebblebis[1],\dots,\pebblebis$ on \dwordbis. Then the Spoiler immediately wins if any of the following equivalences fail, where $\move,\movebis$ are either the interpretations of the same constant symbol in \dword, \dwordbis, or the elements on which $\pebble[j],\pebblebis[j]$ (for some $j\in\{1,\dots,i\}$) have been placed (and similarly for $\movetwo, \movetwobis$):
  \begin{itemize}
  \item $\dword\models\move=\movetwo$ iff $\dwordbis\models\movebis=\movetwobis$
  \item $\dword\models\ProcS(\move)$ (resp. $\ProcE(\move)$) iff $\dwordbis\models\ProcS(\movebis)$ (resp. $\ProcE(\movebis)$)
  \item $\dword\models a(\move)$ iff $\dwordbis\models a(\movebis)$, for $a\in\vocab$
  \item $\dword\models\move\lesssim\movetwo$ iff $\dwordbis\models\movebis\lesssim\movetwobis$
  \item $\dword\models\move\sim\movetwo$ iff $\dwordbis\models\movebis\sim\movetwobis$
  \end{itemize}
  The Duplicator wins if the game reaches the end of round $k$ and if the Spoiler still doesn't win. In that case, we write $\dword\prefEFeq\dwordbis$.

\end{definition}

The prefix \EF game has been tailored to capture the expressive power of \prefFOdw:

\begin{proposition}
  \label{prop:EF}
  Let $k$ be an integer, and let \dword, \dwordbis be two data words on the same alphabet (with constants). Then $\dword\prefEFeq\dwordbis$ if and only if $\dword\prefFOdweq\dwordbis$.
\end{proposition}

\begin{proof}
  We prove this result by induction on $k$. For $k=0$, notice that both $\dword\prefEFeq[0]\dwordbis$ and $\dword\prefFOdweq[0]\dwordbis$ hold exactly when \dword and \dwordbis agree quantifier-free formulas.

  Suppose that this equivalence hold for some integer $k$, and let us consider two data words \dword and \dwordbis with constant set \Const.

  Let us first assume that $\dword\prefEFeq[k+1]\dwordbis$ and show that in that case, \dword and \dwordbis agree on all \prefFOdw of quantifier depth at most $k+1$. Note that every sentence of \prefFOdw of quantifier depth at most $k+1$ is a boolean combination of sentences quantifier-free sentences, and of sentences of the form
  \begin{equation}
    \label{eq:procex}
    \exists x,\ \Proc(x)\ \land\ \formule(\{x\};\emptyset;\emptyset),
  \end{equation}
  \begin{equation}
    \label{eq:boundex}
    \exists x,\ \neg\Proc(x)\ \land\ \Big(\bigwedge\limits_{c\in\boundConst}x\not\sim c\Big)\ \land\ \formule(\emptyset;\{x\};\emptyset) \,,
  \end{equation}
  or
  \begin{equation}
    \label{eq:prefex}
    \exists x,\ x\lesssim c\ \land\ \formule(\emptyset;\emptyset;\{x\})\qquad \text{where }c\in\boundConst\,,
  \end{equation}
  where \formule is a \prefFOdw formula of quantifier depth at most $k$.
  \\
  By assumption the Spoiler does not immediately win, hence \dword and \dwordbis must agree on all quantifier-free sentences. It is thus enough to show that they agree on all sentences of the form (\ref{eq:procex}), (\ref{eq:boundex}) and (\ref{eq:prefex}). Assume that \dword satisfies such an existential sentence, witnessed by some $\move\in\dword$, and let us prove that \dwordbis also satisfies this formula -- the reverse is symmetric. We distinguish between the three cases.
  \\
  Let us consider the case of a sentence of form (\ref{eq:procex}). By assumption, the Duplicator is able to win the $(k+1)$-round game. In particular, if the Spoiler makes a process move by playing \procpebble[1] on \move, then the Duplicator can respond by placing \procpebblebis[1] on \movebis (which must be a process of \dwordbis by definition of the game) and still win $r$ rounds. This means, if we add a new process constant symbol to \Const and interpret it as \move in \dword and \movebis in \dwordbis, yielding respectively \dwordnew and \dwordbisnew, that $\dwordnew\prefEFeq\dwordbisnew$. By induction hypothesis, we thus get $\dword\prefFOdweq\dwordbis$. Then \movebis is a process element witnessing that \dwordbis also satisfies the sentence $\exists x,\Proc(x)\land\formule(\{x\};\emptyset;\emptyset)$.
  \\
  We deal similarly with formulas of type (\ref{eq:boundex}) and (\ref{eq:prefex}), the only distinction being that in the former the Spoiler starts by playing a bounding move, while they play a prefix move in the latter.

  Conversely, let us prove that if \dword and \dwordbis satisfy the same \prefFOdw-sentences of quantifier depth at most $k+1$, then the Duplicator wins the $(k+1)$-round game. Let us assume that the Spoiler's first move is a process move by placing \procpebble[1] on \move in \dword (the case where they play in \dwordbis is fully symmetric). An easy induction on the quantifier depth shows that up to equivalence, there are only a finite number of formulas of a given quantifier depth. Let $\formule(\{x\},\emptyset,\emptyset)$ be the conjunction of all the formulas of quantifier depth at most $k$ satisfied by \move in \dword. By construction, \dword then satifies the sentence $\exists x,\Proc(x)\land\formule(\{x\};\emptyset;\emptyset)$, which is also true in \dwordbis as it has quantifier depth $k+1$. Let $\movebis\in\dwordbis$ be a witness to that fact. The Duplicator will place \procpebblebis[1] on \movebis (note that this is a valid process move, as \movebis must be a process in \dwordbis). Once again, let us enrich the constant set with a new process constant interpreted in \dword and \dwordbis respectively as \move and \movebis, which yield two data words \dwordnew and \dwordbisnew. By choice of \formule we have $\dwordnew\prefFOdweq\dwordbisnew$ and thus, by induction hypothesis, $\dwordnew\prefEFeq\dwordbisnew$, which precisely mean that the Duplicator can win $k$ rounds after the first, and thus can win the $(k+1)$-round games starting with a process move.
  \\
  When the Spoiler starts with a bounding move on an element \move of \dword, we consider a formula $\formule(\emptyset,\{x\},\emptyset)$ characterizing the set of formulas of quantifier depth at most $k$ satisfied by \move in \dword. Then \[\exists x,\neg\Proc(x)\land\Big(\bigwedge\limits_{c\in\boundConst}x\not\sim c\Big)\land\formule(\emptyset;\{x\};\emptyset)\] has quantifier depth $k+1$ and is valid both in \dword and \dwordbis. Any existential witness \movebis in \dwordbis is necessarily a valid option for a bounding move by the Duplicator, who responds by placing \boundpebblebis[1] on \movebis. For the same reason as in the previous case, the Duplicator can the proceed to win $k$ more rounds of the game.
  \\
  If the Spoiler starts the game with a prefix move in which they place \prefpebble[1] on some element \move, then the rules of the game ensure the existence of a bounding constant $c\in\Const$ such that $x\lesssim c$. In this case, we consider the sentence $\exists x,x\lesssim c\land\formule(\emptyset;\emptyset;\{x\})$ (where once again \formule subsumes all formulas of quantifier depth at most $k$ satisfied by \move in \dword) and conclude as before.
  \\
  All in all, the Duplicator wins the $(k+1)$-round game whenever \dword and \dwordbis are such that $\dword\prefFOdweq\dwordbis$, which concludes the proof.
  \ifcsname fullversion\endcsname
  \else
  \qed
  \fi
\end{proof}

On (finite or infinite) words, a straightforward adaptation of this game (where there is no process pebble, and only one bounding pebble) characterizes \prefFO. 

\end{toappendix}

\begin{lemmarep}\label{lemma:threshold}
  Let $k\in\N$ and let \dword and \dwordbis be two finite or infinite data-words on the same alphabet $\vocab$, such that for every $\type\in\Types$, the number of classes of type \type in \dword and \dwordbis are either the same or both at least $k$. Then \dword and \dwordbis agree on all \prefFOdw-sentences of quantifier depth at most $k$.
\end{lemmarep}

\begin{appendixproof}
  To prove this result, we use the previously introduced \EF games on data-words. In view of Proposition~\ref{prop:EF}, it is enough to show that given $k\in\N$, if for each $k$-type \type for \prefFO \dword and \dwordbis have the same number (up to threshold $k$) of classes of type \type, then $\dword\prefEFeq\dwordbis$.

  Under this assumption, we are going to show that the Duplicator can play in such a way as to preserve the invariant described in the remainder of this paragraph. Let us assume that pebbles $\pebble[1],\dots,\pebble[i]$ and $\pebblebis[1],\dots,\pebblebis[i]$ have been placed respectively on \dword and \dwordbis. Define the equivalence relation $R$ on $\{\pebble[1],\dots,\pebble[i]\}$ (resp. $\widehat R$ on $\{\pebblebis[1],\dots,\pebblebis[i]\}$) where $\move R\movetwo$ iff $\dword\models\move\sim\movetwo$ (resp. $\movebis\widehat R\movetwobis$ iff $\dwordbis\models\movebis\sim\movetwobis$). The first part of the invariant is that the bijection sending $\pebble[j]$ to $\pebblebis[j]$ maps $R$ to $\widehat R$. For every equivalence class $c=\{\pebble[j_1],\dots,\pebble[j_l]\}$ of $R$ (and the corresponding class $\{\pebblebis[j_1],\dots,\pebblebis[j_l]\}$ of $\widehat R$), let $w_c$ be the class (seen as a word) of \dword on which $\pebble[j_1],\dots,\pebble[j_l]$ have been placed, together with one constant for each bounding or prefix pebble among $\pebble[j_1],\dots,\pebble[j_l]$ (there is no need to remember process pebbles) placed at the same position as the corresponding pebble in \dword, and let $\widehat w_c$ be defined similarly and on the same vocabulary, but for \dwordbis. The second part of the invariant is that $w_c\prefEFeq[k-i]\widehat w_c$.

  Let us describe how the Duplicator can enforce this invariant during the play, by induction on the number $i$ of rounds. For $i=0$, there is nothing to show. Assuming the invariant holds after $i$ round, let us consider (without loss of generality) the case where the Spoiler places a pebble \pebble[i+1] among $\{\procpebble[i+1],\boundpebble[i+1],\prefpebble[i+1]\}$ in \dword. We distinguish between two situations:
  \begin{itemize}
  \item Let us first assume \pebble[i+1] is placed on a class where some previous pebble has already been placed, and let $c$ be the corresponding equivalence class of $R$. Our inductive assumption ensures $w_c\prefEFeq[k-i]\widehat w_c$. If $\pebble[i+1]=\procpebble[i+1]$, the Duplicator simply places \procpebblebis[i+1] on the process element of \dwordbis corresponding to $\widehat w_c$, and the invariant for $i+1$ obviously holds. Assume now $\pebble[i+1]=\boundpebble[i+1]$ (resp. \prefpebble[i+1]) has been placed on element \move. Remember that $w_c\prefEFeq[k-i]\widehat w_c$: then the Duplicator in the game between \dword and \dwordbis answers by placing \boundpebblebis[i+1] (resp. \prefpebblebis[i+1]) on \movebis, where \movebis is what the winning strategy for the Duplicator in the game between $w_c$ and $\widehat w_c$ answers when the Spoiler makes their move on \move. The invariant after round $i+1$ follows easily, since $w'_c\prefEFeq[k-i-1]\widehat w'_c$, where $w'_c$ (resp. $\widehat w'_c$) is the extension of the word $w_c$ with a new constant, interpreted as \move (resp. \movebis).
  \item Suppose now that \pebble[i+1] is placed on a class $w$ of \dword on which no previous pebble has been placed. By assumption about \dword and \dwordbis sharing the same number of classes of every type up to threshold $k$, we can find a class $\widehat w$ of \dwordbis which has the same $k$-type for \prefFO as $w$ (i.e. such that $w\prefEFeq\widehat w$), and on which no previous pebble has been placed either. If $\pebble[i+1]=\procpebble[i+1]$, then the Duplicator places \procpebblebis[i+1] on the process element of $\widehat w$, and $w\prefEFeq\widehat w$ in particular entails $w\prefEFeq[k-i-1]\widehat w$. On the other hand, if $\pebble[i+1]=\boundpebble[i+1]$ has been placed on element \move of $w$, then $w\prefEFeq\widehat w$ means that the Duplicator can place \boundpebblebis[i+1] on some element \movebis of $\widehat w$ such that $w'\prefFOeq[k-1]\widehat w'$ (and, a fortiori, $w'\prefFOeq[k-i-1]\widehat w'$) where $w'$ (resp. $\widehat w'$) is the same as $w$ (resp. $\widehat w$) where a new constant is interpreted as \move (resp. \movebis). The invariant after round $i+1$ follows.
  \end{itemize}
  It is straightforward to notice that enforcing this invariant after the $k$ rounds of the game guarantees a win for the Duplicator.
\ifcsname fullversion\endcsname
\else
\qed
\fi
\end{appendixproof}

As a direct corollary of Lemma~\ref{lemma:threshold}, in order to decide whether a data word \dword satisfies an \prefFOdw formula of quantifier depth $k$, it is enough to know, for each $k$-type \type for \prefFO, how many (up to $k$) classes of \dword have type \type.
We shall use this fact later in proofs, and refer to this abstraction of a data word as its \emph{collection} of types.

\subsection{Properties of \prefFO types}\label{sec:lemmas}

Let us now try to understand the behavior of \prefFO on words.
In the following, we fix an alphabet \vocab and an integer $k$.

First, note that the equivalence relation ``have the same \prefFO-type'' is not a congruence in the monoid of finite words. Indeed, one can convince themselves (or prove formally, using the \EF games introduced in the
\ifcsname fullversion\endcsname
appendix
\else
full version of this paper
\fi
) that for any $k\in\N$, the two words
\[u=ababa\cdots aba\]
and
\[v=ababa\cdots abab\]
have the same \prefFO $k$-type (as long as they are long enough with respect to $k$). However, $u\cdot a$ and $v\cdot a$ can be separated by the \prefFO-sentence of quantifier depth $3$ stating the existence of two consecutive $a$'s.

\begin{toappendix}
This easy technical lemma will prove most useful in the following:

\begin{lemma}
  \label{lem:subgame}
  Let $u$ and $u'$ be two words on \vocab, and $n,n'\in\N$ be such that such that $\tpword[u[0\dots n]] = \tpword[u'[0\dots n']]$.

  In the $k$-round prefix \EF game between $u$ and $u'$, if the Spoiler's first move is in $u[0\dots n]$ or in $u'[0\dots n']$, then the Duplicator wins the game.
\end{lemma}

\begin{proof}
  By assumption, the Spoiler plays their first move (which must be a bounding move) in $u[0\dots n]$ or $u'[0\dots n']$. The Duplicator responds as they would in the $k$-round prefix \EF game between $u[0\dots n]$ and $u'[0\dots n']$.

    The bounding pebbles are thus in the prefixes of $u$ and $u'$ for which the Duplicator has a winning strategy, and have been placed according to this strategy. By definition of the game, no pebble will ever be placed outside of these prefixes. The Duplicator can thus follow the same strategy to win the game between $u$ and $u'$ where the Spoiler has played their first move in $u[0\dots n]$ or $u'[0\dots n']$.    
    \ifcsname fullversion\endcsname
    \else
    \qed
    \fi
\end{proof}

When one word is the prefix of the other, we immediately get the following consequence:

\begin{corollary}
  \label{cor:prefix_win}
  Let $u$ and $v$ be two words such that $u$ is a prefix of $v$, and let us consider the prefix \EF game between $u$ and $v$.

  If the Spoiler's first move is in $u$, then the Duplicator wins the game, no matter the number of rounds.
\end{corollary}
\end{toappendix}

\begin{lemmarep}
  \label{lem:sandwich}
  Let $u,v$ be words such that $\tpword[u] = \tpword[u \cdot v] = \type$. Then for every prefix $w$ of $v$, $\tpword[u\cdot w] = \type$.
\end{lemmarep}

\begin{appendixproof}
  In order to prove this result, we fix such a $w$ and show that the Duplicator has a winning strategy in the $k$-round prefix \EF game between $u$ and $u\cdot w$.

  If the Spoiler first plays in $u$, then Lemma~\ref{lem:subgame} guarantees the Duplicator wins the game. Thus, we can assume that the Spoiler's first move is a bounding move in $u\cdot w$. By Lemma~\ref{lem:subgame}, the Duplicator wins the $k$-round game between $u\cdot w$ and $u\cdot v$ where the Spoiler starts by playing in $u\cdot w$. To win the game between $u$ and $u\cdot w$ when Spoiler starts in $u\cdot w$, the Duplicator combines their strategy in the game between $u$ and $u\cdot v$ and their strategy in the game between $u\cdot v$ and $u\cdot w$ where Spoiler starts in $u\cdot w$.
  \ifcsname fullversion\endcsname
  \else
  \qed
  \fi
\end{appendixproof}

This lemma has a straightforward consequence in the case of infinite words: for every infinite word $u$, there exists some $\type\in\Types$ and an index $n\in\N$ such that for every $m\geq n$, $u[0\dots m]$ has type \type. Indeed, \Types is finite, thus there must be some type appearing infinitely often in the prefixes of $u$. Lemma~\ref{lem:sandwich} ensures that such a type is unique. We refer to this type as the \emph{stationary type} of $u$.

Next, we prove that the stationary type of an infinite word is none other than its own type:

\begin{lemmarep}\label{lemma:stat type prelim}
  Let $u$ be an infinite word with stationary type $\type\in\Types$. Then $\tpword[u] = \type$.
\end{lemmarep}

\begin{appendixproof}
  Let $n$ be such that $u[0\dots n]$ has type \type. We show that the Duplicator has a winning strategy in the $k$-round prefix \EF game between $u$ and $u[0\dots n]$.

  Due to Corollary~\ref{cor:prefix_win}, the Spoiler cannot win by playing their first move in $u[0\dots n]$. Let us therefore assume that the Spoiler plays first on position $m$ of $u$. By assumption, there exists a position $m'>m$ such that $u[0\dots m']$ has type \type, i.e. $u[0\dots n]\prefFOeq u[0\dots m']$. Combining the winning strategies for the Duplicator in the game between $u[0\dots n]$ and $u[0\dots m']$, and in the game between $u[0\dots m']$ and $u[0\dots m]$ where the Spoiler plays first in $u[0\dots m]$ (following from Corollary~\ref{cor:prefix_win}), the Duplicator can win the $k$-round prefix \EF game between $u[0\dots n]$ and $u[0\dots m]$ when the Spoiler plays first in $u[0\dots m]$.

  All in all, $u\foeq u[0\dots n]$, i.e. $u$ has type \type.
  \ifcsname fullversion\endcsname
  \else
  \qed
  \fi
\end{appendixproof}

Combining those results we get the following:

\begin{corollary}\label{lemma:stationary type}
Let $k \in \N$, $u$ be an infinite word, and $\type=\preffotpword[u]$.
Then there exists $n \in \N$ such that for all $m > n$, $u[0 \dots m]$ also has type \type.
\end{corollary}

Let us now try to understand the structure of \Types.
We consider the binary relation \preordtp defined on \Types as follows: $\type\preordtp\typebis$ if and only if there exists two finite words $u$ and $v$ such that $\tpword[u]=\type$ and $\tpword[uv]=\typebis$.
We refer to $(\Types,\preordtp)$ as the \emph{graph of \prefFO-types} of depth $k$.

The following lemmas break down its properties.
First, the choice of $u$ and $v$ above does not really matter:

\begin{lemmarep}
  \label{lem:preordtp_forall}
  Let $k\geq 1$, let \type and \typebis be such that $\type\preordtp\typebis$ and let $w$ be a finite word such that $\tpword[w]=\type$. There exists some finite word $w'$ such that $\tpword[ww']=\typebis$.
\end{lemmarep}

\begin{appendixproof}
  By definition, there exists some finite words $u,v$ such that $\tpword[u]=\type$ and $\tpword[uv]=\typebis$.

    We can assume that $u\neq\motvide$, for otherwise $w=u$ and there is nothing to show.

    The Duplicator has a winning strategy in the $k$-round prefix \EF game between $u$ and $w$. Let $n$ be the position in $u$ of their response when the Spoiler plays a bounding move on the last position of $w$ (which we assumed exists). Let us call \strat the rest of Duplicator's strategy when these first bounding moves have been played: \strat thus allows the Duplicator to win for $k-1$ more prefix moves made in $w$ and $u[1\dots n]$.

    With this strategy fixed, let us show that \[w':=u[n+1\dots]\cdot v\] is a good candidate for our purpose, i.e. that $\tpword[w\cdot w']=\typebis$ for such a choice of $w'$. To that end, we describe a winning strategy for the Duplicator in the $k$-round prefix \EF game between $u\cdot v$ and $w \cdot u[n+1\dots]\cdot v$.

    If the Spoiler makes their bounding move in the prefix $u$ of $u\cdot v$ or in the prefix $w$ of $w\cdot u[n+1\dots]\cdot v$, Lemma~\ref{lem:subgame} ensures the Duplicator wins the game.
    
    Let us now assume that the Spoiler plays their bounding move in the suffix $u[n+1\dots]\cdot v$ of $w\cdot u[n+1\dots]\cdot v$ (resp. in the suffix $v$ of $uv$). In this case, the Duplicator answers by playing on the corresponding element of the suffix $u[n+1\dots]\cdot v$ of $uv$ (resp. the corresponding element of the suffix $v$ of $w\cdot u[n+1\dots]\cdot v$). Following this, the Duplicator's strategy is
    \begin{itemize}
    \item to play tit-for-tat whenever the Spoiler plays in either of the suffixes $u[n+1\dots]\cdot v$, and
    \item to follow \strat whenever the Spoiler plays in the prefix $w$ of $w\cdot u[n+1\dots]\cdot v$ or in the prefix $u[1\dots n]$ of $uv$.
    \end{itemize}
    By definition of \strat, this strategy allows the Duplicator to win for $k-1$ additional rounds.

  We have covered every case, which all lead to the Duplicator winning the $k$-round prefix game between $u\cdot v$ and $w\cdot u[n+1\dots]\cdot v$, thus showing that $\tpword[ww']=\typebis$ for $w':=u[n+1\dots]\cdot v$.
  \ifcsname fullversion\endcsname
  \else
  \qed
  \fi
\end{appendixproof}

\begin{figure}
  \centering
  \begin{tikzpicture}
    [scale=2,
      every place/.style={draw=gray,fill=gray!10,minimum size=2cm},
      every token/.style={draw=gray,fill=gray!10,minimum size=1mm}
    ]
    \node[place,
    ] (vide) at (0,0) {$\smallpreffotpword[\motvide][2]$};
    \node[place,
    ] (aE) at (-1,-.9)  {$\smallpreffotpword[\letterE][2]$};
    \node[place,
    ] (aEaE) at (-1,-2.2)  {$\smallpreffotpword[\letterE\letterE][2]$};
    \node[place,
    ] (aS) at (1,-.9)  {$\smallpreffotpword[\letterS][2]$};
    \node[place,
    ] (aSaS) at (1,-2.2) {$\smallpreffotpword[\letterS\letterS][2]$};
    \foreach \s/\t/\dir in {vide/aE/right to,aE/aEaE/right to,vide/aS/left to,aS/aSaS/left to} 
    \draw (\s) edge[gray,-\dir,thick] (\t);
  \end{tikzpicture}
  \caption{A partial representation of \Types[2] for $\vocabE=\{\letterE\}$ and $\vocabS=\{\letterS\}$. Types of words containing both \letterE and \letterS have been omitted, as data classes cannot have such a type.}
  \label{fig:types}
\end{figure}

Second, \preordtp is an order:
\begin{lemmarep}\label{lemma: preordtp order}
  The binary relation \preordtp is a partial order on \Types, with minimum \preffotpword[\motvide].
\end{lemmarep}

\begin{appendixproof}
  Reflexivity is straightforward, and so is the fact that for every $\type\in\Types$, $\preffotpword[\motvide]\preordtp\type$.

    To check that \preordtp is transitive, let us consider three types such that $\type\preordtp\typebis\preordtp\typeter$. By definition, there exists two finite words $u$ and $u'$ such that $\preffotpword[u]=\type$ and $\preffotpword[uu']=\typebis$. Lemma~\ref{lem:preordtp_forall} entails the existence of some finite word $u''$ such that $\preffotpword[uu'u'']=\typeter$, which concludes the proof.

  It only remains to check the antisymmetry of \preordtp. Let $\type,\typebis\in\Types$ be such that $\type\preordtp\typebis$ and $\typebis\preordtp\type$. By Lemma~\ref{lem:preordtp_forall}, there must exist three finite words $u,u',u''$ such that $\preffotpword[u]=\preffotpword[uu'u'']=\type$ and $\preffotpword[uu']=\typebis$. From there, Lemma~\ref{lem:sandwich} derives the equality of \type and \typebis.
  \ifcsname fullversion\endcsname
  \else
  \qed
  \fi
\end{appendixproof}

From those two lemmas, we conclude that $(\Types,\preordtp)$ can be seen as a finite directed tree rooted in $\preffotpword[\motvide]$. This is illustrated in Figure~\ref{fig:types}.

\section{Synthesis and token games}\label{sec:synthesis}
In this section we define the standard synthesis game, and then give equivalent games that are more suitable for our purpose.

\subsection{Standard synthesis game}
\label{sec:standard_game}
Given a formula \formule, the \emph{standard synthesis game} is a game played between two players, \Sys and \Envi, who collaborate to create a data word.
\Sys's goal is to make the created data word satisfy \formule, while \Envi wants to falsify it.
Formally, a \emph{strategy} for \Sys is a function $\strat: \DW \to (\vocabS \times \ProcS) \cup \{\varepsilon\}$ which given a data word created so far (the \emph{history}) returns either an action and process to play on, or passes its turn on output $\varepsilon$.
A data word $\dw = (a_0,p_0) (a_1,p_1) \dots$ is \emph{compatible} with $\strat$ if for all $i$, $a_i \in \vocabS$ implies $\strat(\dw[][][0 \dots i-1]) = (a_i,p_i)$.
Furthermore, $\dw$ is \emph{fair} with $\strat$ if either $\dw$ is finite and $\strat(\dw) = \varepsilon$, or $\strat(\dw[][][0 \dots i]) \neq \varepsilon$ for infinitely many $i \in \N$ implies $a_i \in \vocabS$ for infinitely many $i \in \N$.
Intuitively, a fair data word prevents the pathological case where \Sys wants to do some action but \Envi forever prevents it by continually playing its own actions instead.
Strategy $\strat$ is said to be \emph{winning} if all compatible and fair data words satisfy \formule.
\Sys wins a synthesis game if there exists a winning strategy for \Sys.

The \emph{existential synthesis problem} asks, for a given alphabet \vocab and formula \formule, whether there exists a set of processes $\Proc=\ProcS\uplus\ProcE$ such that \Sys wins the corresponding synthesis game.
In the case of \prefFOdw, since the logic can only compare process identities with respect to equality, it is easy to see that the actual sets \ProcS and \ProcE do not matter: only their cardinality does.
With that in mind, we slightly reformulate the existential synthesis problem to ask whether there exists a pair $(\nS, \nE) \in \N^2$ such that \Sys wins the synthesis game for all sets of processes $\ProcS$ and $\ProcE$ of respective size \nS and \nE.
If $(\nS, \nE)$ is such a pair, we say that \Sys has a $(\nS, \nE)$-winning strategy for \formule.

\subsection{Symmetric game}
Notice that the standard synthesis game is asymmetric: while \Envi can interrupt \Sys at any point (provided that \Sys gets the possibility to play infinitely often if they want to), \Sys does not choose exactly the timing of their moves. In particular, \Sys is never guaranteed to be able to play successive moves in the game. Let us define a variation of the game, which turns out to be equivalent with respect to the logic \prefFOdw, in which \Sys can play arbitrarily many successive moves. This makes the game symmetric, and will make the following proofs simpler.

The \emph{symmetric game} is defined in the same way as the standard game, with the following exceptions.
Instead of a function from $\DW$ to $(\vocabS \times \ProcS) \cup \{\varepsilon\}$, a strategy for \Sys is now a function from $\DW$ to $(\vocabS \times \ProcS)^\star \cup \{\varepsilon\}$, i.e. \Sys is allowed to play an arbitrary (but finite) amount of actions at once.
Moreover, players strictly alternate, starting with \Envi. 
The rest is defined as in the standard game.
It is obvious that \Sys has an easier time winning the symmetric game than the standard game. It turns out the symmetric game offers \Sys no advantage when the relative positions of the classes are incomparable:

\begin{lemmarep}\label{lemma:symmetric}
  Let $\formule$ be a \prefFOdw-sentence, and $\nS,\nE\in\N$. \Sys has an $(\nS,\nE)$-winning strategy for $\formule$ in the symmetric game if and only if \Sys has an $(\nS,\nE)$-winning strategy for $\formule$ in the standard game.
\end{lemmarep}

\begin{appendixproof}
It is straightforward to notice that if \Sys has a winning strategy for the standard game, then the same strategy is winning for the symmetric game. 

For the reverse direction, we need an additional observation regarding the logic \prefFOdw.
In general, given a data word $\dw$, in any reordering of the positions of $\dw$ such that the class of each process is left unchanged, then the satisfaction value of any formula in \prefFOdw is the same as for the original word.
Indeed, since the satisfaction of \prefFOdw only depends on the collection of types of processes and not their relative ordering, a \prefFOdw formula cannot distinguish between the two data words.
In other terms, we can freely swap any two consecutive positions in a data words as long as they do not share the same process, while keeping the same satisfaction value for any \prefFOdw formula.
Formally, we say that two data words $\dw[1]$ and $\dw[2]$ are \emph{equivalent}, denoted by $\dw[1] \equiv \dw[2]$, if $\dw[2]$ is a reordering of $\dw[1]$ and for all processes $p$, the class of $p$ in $\dw[1]$ and $\dw[2]$ are the same.
As explained above, if two data words are equivalent, then they satisfy exactly the same \prefFOdw formulas; formally, one shows that the natural strategy for the Duplicator in the prefix \EF game at any depth (which consists of responding to a Spoiler's move by playing on the same position of the same class) is winning.

Now, let us assume that \Sys has a $(\nS,\nE)$-winning strategy $\strat$ in the symmetric game.
We build a strategy $\strat'$ in the standard game using additional memory $m(\dw) = (\dw[\sys], \dw[\env], \dw[\equiv]) \in {(\DW)}^3$ that depends on the history $\dw$.
Intuitively, $\strat'$ tries to mimic $\strat$, and to this end $\dw[\sys]$ is a word played by $\strat$ that $\strat'$ tries to play bit by bit.
However \Envi actions can happen in between \Sys actions.
Those are then stored in $\dw[\env]$, a queue of actions made by \Envi to be acknowledged later by $\strat'$.
That is, when $\strat'$ finishes playing the word that was stored in $\dw[\sys]$ and needs to consult $\strat$ to see what is the next word to be played, then all of $\dw[\env]$ is added to the history as if they had just happened at this point.
Finally, $\dw[\equiv]$ records a history in the symmetric game such that $\dw[\equiv] \cdot \dw[\env]$ is equivalent to the current history $\dw$ in the standard game and such that $\dw[\equiv] \cdot \dw[\sys]$ is a $\strat$-compatible history.

First let us define $\strat'$. On a given history $\dw \in \DW$ and with memory $m(\dw) = (\dw[\sys],\dw[\env],\dw[\equiv])$, $\strat'$ is defined as follows.
If $\dw[\sys] = (a_0, p_0) \dots$ is not empty, then $\strat'(\dw) = (a_0,p_0)$.
Otherwise $\dw[\sys] = \motvide$, and with $\strat(\dw[\equiv] \cdot \dw[\env]) = (a_0,p_0) \dots$ we let $\strat'(\dw) = (a_0,p_0)$.
If $\strat(\dw[\equiv] \cdot \dw[\env]) = \varepsilon$ instead, then $\strat'(\dw) = \varepsilon$ as well.
Let us now show how the memory is updated after each action.
The initial memory is $m(\motvide) = (\motvide, \motvide, \motvide)$.
Suppose now that the current history is $\dw$ and the current memory is $m(\dw) = (\dw[\sys],\dw[\env],\dw[\equiv])$.
On an \Envi action $(a,p)$, we update the memory to $m(\dw \cdot (a,p)) = (\dw[\sys],\dw[\env] \cdot (a,p),\dw[\equiv])$.
On a \Sys action $(a,p)$, if $\dw[\sys] = (a,p) \cdot \dw[\sys][-]$ then $m(\dw \cdot (a,p)) = (\dw[\sys][-],\dw[\env] ,\dw[\equiv] \cdot (a,p))$.
If $\dw[\sys] = \motvide$, and $\strat(\dw[\equiv] \cdot \dw[\env]) = (a,p) \cdot \dw[][\strat]$, then we let $m(\dw \cdot (a,p)) = (\dw[][\strat],\motvide ,\dw[\equiv] \cdot \dw[\env] \cdot (a,p))$.
Otherwise $m(\dw \cdot (a,p))$ is undefined.

We need a few properties regarding $m$ before showing that $\strat'$ is winning.
Let $\dw$ be a history in the standard game and let $m(\dw) = (\dw[\sys],\dw[\env],\dw[\equiv])$.
First we show by induction that $\dw \equiv \dw[\equiv] \cdot \dw[\env]$.
The initial case is trivial.
Suppose now that $\dw \equiv \dw[\equiv] \cdot \dw[\env]$, and let $(a,p)$ be the next action with $m(\dw \cdot (a,p)) = (\dw[\sys][2],\dw[\env][2],\dw[\equiv][2])$.
If $a$ is an \Envi action or if $\dw[\sys] = \motvide$, then we instantly get that $\dw \cdot (a,p) \equiv \dw[\equiv][2] \cdot \dw[\env][2]$ by the definition of $m$ and the induction hypothesis.
If $a$ is a \Sys action and $\dw[\sys] \neq \motvide$, then notice that since $\dw[\env]$ entirely consists of \Envi actions, none of those involve process $p$ which is a \Sys process.
Thus $\dw[\equiv] \cdot (a,p) \cdot \dw[\env] \equiv \dw[\equiv] \cdot \dw[\env] \cdot (a,p)$, and the rest follows.
The facts that $\dw[\equiv] \cdot \dw[\sys]$ is always $\strat$-compatible and that $m(\dw)$ is always defined on any $\strat'$-compatible history are also straightforward to prove inductively.

Finally, it remains to be shown that $\strat'$ is winning in the standard game.
Let $\dw$ be a $\strat'$-compatible $\strat'$-fair history in the standard game.

If $\dw$ is finite, then necessarily $\strat'(\dw) = \varepsilon$. 
The only possible way for this to occur is that, with $m(\dw) = (\dw[\sys],\dw[\env],\dw[\equiv])$, we have that $\dw[\sys] = \motvide$ and $\strat(\dw[\equiv] \cdot \dw[\env]) = \varepsilon$.
In that case, since $\strat$ is winning then $\dw[\equiv] \cdot \dw[\env]$ must satisfy \formule. Since $\dw \equiv \dw[\equiv] \cdot \dw[\env]$, then so does $\dw$.

If $\dw$ is infinite, then either $\strat'$ only outputs $\varepsilon$ from some point on, in which case the previous argument also shows that $\dw$ satisfies \formule, or $\strat'$ outputs some action infinitely often.
By the fairness constraint, it means that \Sys actions have been made infinitely often in $\dw$.
Let $m(\dw[][][0 \dots i]) = (\dw[\sys][i], \dw[\env][i], \dw[\equiv][i])$ denote the memory after the $i$-th first actions in $\dw$.
Let $\dw[\equiv]$ be the limit of the sequence of $\dw[\equiv][i] \cdot \dw[\env][i]$, which is well-defined by definition of the memory.
Since for all $i$ the ``buffer'' $\dw[\sys][i]$ only contains a finite amount of \Sys actions to be played by $\strat'$, it follows that $\dw[\sys][i] = \motvide$ for infinitely many $i$.
At those points, we have that $\dw[][][0 \dots i] \equiv \dw[\equiv][i] \cdot \dw[\env][i]$ and that $\dw[\equiv][i] \cdot \dw[\env][i]$ is $\strat$-compatible, as $\dw[\sys][i]$ is empty and adding \Envi actions to $\dw[\equiv][i]$ does not change its $\strat$-compatibility.
Thus, $\dw[\equiv]$ is an infinite $\strat$-compatible data word, and for infinitely many $i$ we have that $\dw[][][0 \dots i] \equiv \dw[\equiv][][0 \dots i]$.
From this, we deduce that $\dw \equiv \dw[\equiv]$.
Indeed if that was not the case then either $\dw$ is not a reordering of $\dw[\equiv]$, which means that there is some finite position $i$ in $\dw$ or $\dw[\equiv]$ that does not have a counterpart in the other data word, contradicting the equivalence between the next two prefixes.
Or, there is a class $p$ in $\dw$ that differs from $\dw[\equiv]$, which again means that they differ on some finite position in that class that corresponds to some position $i$ in $\dw$ and $\dw[\equiv]$, again contradicting the infinitely many prefixes that are equivalent.
To sum up, $\dw$ is equivalent to $\dw[\equiv]$, and $\dw[\equiv]$ satisfies \formule because $\strat$ is winning in the symmetric game, therefore $\dw$ satisfies \formule and thus $\strat'$ is winning.
\ifcsname fullversion\endcsname
\else
\qed
\fi
\end{appendixproof}

Note that when positions between classes can be compared, the standard and the symmetric game do not necessarily agree on who wins for a given formula: consider for instance the formula making \Sys the winner if they manage to play twice in a row: the symmetric game is easily won by \Sys, but is won by \Envi in the standard setting.

\subsection{Token game on words}
Our goal is to give an alternative game played on a finite arena that would still be equivalent to the symmetric game.
As an intermediate step, consider the (infinite) graph of all finite words over $\vocab$: $\arena_\vocab = (\vocab^\star, \motvide, \Transi)$ where $\vocab^\star$ is the set of nodes, $\motvide$ is the initial node and $\Transi \subseteq \vocab^\star \times \vocab \to \vocab^\star$ is the transition function such that $\Transi(w, a) = w \cdot a$.
We use this graph as an arena over which a number of tokens are located, each token representing one process and its location being the history of what has been played on this process.
We have two sets of \Sys and \Envi tokens, numbering $\nS$ for \Sys tokens and $\nE$ for \Envi tokens, all of which are initially placed in the $\motvide$ node.
Alternatively and starting from \Envi, each player picks one of their token, move it along one edge, and repeat those two operations a finite amount of times.
A player can also opt not to move any of its tokens.
Then the other player does the same, and this goes on forever.
The winning condition is defined by the formula \formule: given a play, we take the limit word reached by each token, and see if the collection of those words satisfy \formule.
It is easy to see that this game is simply a different view of the symmetric game: playing a word $w = (a_0, p_0) \dots (a_n, p_n)$ is equivalent to picking the token representing $p_0$, moving it along the $a_0$-labeled edge, and so on.
Thus \Sys wins in the symmetric game if and only if \Sys wins the token game on words, for any choice of token sets of correct sizes.

\subsection{Token game on types}
We already established as a consequence of Lemma~\ref{lemma:threshold} that we do not actually need to keep track of the full history of each token to know whether \formule (which has depth $k$) is satisfied; counting how many tokens (up to threshold $k$) there are for each $k$-type of \prefFO is enough to decide.
Therefore, the final step is to use the finite graph of \prefFO-types as an arena instead of using the infinite graph of all words.
As for the acceptance condition, the formula \formule is abstracted by a set $\Acc$ of $k$-counting functions of the form $\counting: \Types \to \{0,1,\dots,k-1,k+\}$.
Each such function gives a count of how many tokens (up to $k$) can be found in each type, and the acceptance condition $\Acc$ is exactly the set of those functions that satisfy \formule.

Let the \prefFO-arena of depth $k$ be $\arena_k = (\Types, \tpword[\motvide], \preordtp)$ where \Types is the set of nodes, \tpword[\motvide] is the initial node, and \preordtp is the transition function as defined in Section~\ref{sec:lemmas}.
Given a pair $(\nS,\nE) \in \N^2$, let us fix two arbitrary disjoint sets of \Sys and \Envi \emph{tokens} \TokS and \TokE of sizes \nS and \nE respectively, and let \Tok denote their union.
The \emph{token game} over \Tok for a \prefFOdw formula \formule of depth $k$  is given by the tuple $\game_\formule^\Tok = (\Tok, \arena_k, \Acc)$.

A \emph{configuration} of this game is a mapping $\conf: \Tok \to \Types$ indicating where each token lies in the arena.
The initial configuration \confinit maps every token to the initial type $\type_0 = \tpword[\motvide]$.
Starting from \Envi, players alternatively pick a number of their respective tokens and move them in the arena following transitions from \preordtp.
A \emph{move} for \Sys (resp. \Envi) is a mapping $\moveS: \TokS \to \Types$ (resp. $\moveE: \TokE \to \Types$) indicating where to move each token such that for all $\tok \in \TokS$, $\conf(\tok) \preordtp \moveS(\tok)$ (and similarly for \Envi).
In particular, in a given configuration \conf, an empty \Sys move is simply a move equal to \conf restricted to \TokS, indicating that all \Sys tokens should stay where they are.

Then a \emph{play} \play is a sequence of configuration and moves starting from an \Envi move and alternating players: $\play = \confinit \xrightarrow{\moveE^0} \conf_1 \xrightarrow{\moveS^1} \conf_2 \xrightarrow{\moveE^2} \cdots$ such that each new configuration is the result of applying the previous move to the previous configuration.
Let $\Plays$ denote the set of plays.
A play is \emph{maximal} if it is infinite.

For a given token $\tok \in \Tok$, a maximal play \play generates a sequence of types $\tpword[\motvide] = \type_0 \preordtp \type_1 \preordtp \dots$ such that $\type_i = \conf_{2i}(\tok)$.
Note that it is fine to skip every other configuration, as a token can only be moved during either a \Sys or \Envi move but not both.
This infinite sequence eventually loops in some type $\type_\tok$ forever due to the graph of types being a finite tree.
The \emph{limit configuration} for \play, denoted by $\conf_\infty^\play$, is the configuration that returns $\type_\tok$ for every token $\tok$.
Note that there must exists some $i \geq 0$ such that for all $j > i$, $\conf_j =  \conf_\infty^\play$.
We slightly abuse notations and denote by $\play(\tok)$ the type of token $\tok$ in either the last configuration of \play if it is finite or its limit configuration if it is infinite.
A play is \emph{winning} if its limit configuration satisfies the acceptance condition $\Acc$, that is if there is a function $\counting \in \Acc$ such that for all $\type \in \Types$, 
\[
\begin{cases}
  |\{\tok \in \Tok \mid \play(\tok) = \type\}| = \counting(\type)&\text{ if }\counting(\type) < k\,,\text{ and}\\
  |\{\tok \in \Tok \mid \play(\tok) = \type\}| \geq k&\text{ if }\counting(\type) = k+\,.
\end{cases}
\]

A \emph{strategy} for \Sys is a function \strat that given a play returns a \Sys move.
A play is \emph{compatible} with \strat if all \Sys moves in that play are those given by \strat.
A strategy for \Sys is winning if all maximal plays compatible with it are winning.
Finally, we say that a pair $(\nS,\nE) \in \N^2$ is winning for \Sys if \Sys has a winning strategy in the token game (for any choice of token sets \TokS and \TokE of corresponding sizes).

\begin{lemmarep}\label{lemma:token game}
  A pair $(\nS,\nE)$ is winning for \Sys in the token game for \formule if and only if \Sys has a $(\nS,\nE)$-winning strategy for \formule in the standard synthesis game.
\end{lemmarep}

\begin{appendixproof}
Winning in the token game on types is the same as winning in the token game on words: a strategy in the word arena moving a token from word $w$ to $w \cdot w'$ is mimicked in the type arena by the strategy moving the same token from $\tpword[w]$ to $\tpword[w \cdot w']$ (by definition of \preordtp, this is a valid move).
Conversely, if a strategy in the type arena wants to move a token from $\type$ to $\typebis$, then by Lemma~\ref{lem:preordtp_forall} we know that whatever word $w$ such that $\tpword[w] = \type$ the token is in at the time, there is an extension $w'$ such that $\tpword[ww'] = \typebis$ that can be played by the strategy in the arena of words.
Furthermore, the acceptance conditions are equivalent: by Lemma~\ref{lemma:stationary type} the type of some token $\tok$ in the limit configuration of a play is exactly the type of the data word of the corresponding process.
Thus, the acceptance condition of the token game in the type arena is equivalent to winning in the word arena.
Finally, winning in the token game on words is equivalent to winning in the symmetric game, which itself is equivalent to winning in the standard synthesis game, thus \Sys wins the token game on types if and only if \Sys wins the standard synthesis game.
\ifcsname fullversion\endcsname
\else
\qed
\fi
\end{appendixproof}

For a fixed pair $(\nS,\nE)$, the token game is a finite, albeit very large, game.
Remember that our goal is to find whether there exists such a pair that is winning.
We show in the next section how to reduce the search to a (large but) finite space.

\section{Double cutoff for solving the synthesis problem}\label{sec:decidability}
Recall that in terms of expressive power, \prefFOdw is located somewhere between \FOeq (whose existential synthesis problem is known to be decidable~\cite{grange2023first}) and \FOeqord (for which is it undecidable already when restricting to two variables, i.e. for \FOtwoeqord~\cite{grange2023first}).

In this section, we make a step towards closing the gap by proving our main result:

\begin{theorem}
  \label{th:main}
  The existential synthesis problem for \prefFOdw is decidable.
\end{theorem}

To prove Theorem~\ref{th:main} we follow a double cutoff strategy. We first show in Section~\ref{sec:env cutoff} that there is no point in considering too many tokens for \Envi, where the bound depends on the quantifier depth $k$ of the formula \formule but, importantly, not on the number of \Sys tokens. As a second step, we prove in Section~\ref{sec:sys_cutoff} that given a fixed number of \Envi tokens (which is a reasonable assumption in view of the previous point), one can restrict one's study to a space where the number of \Sys tokens is bounded by a function of $k$ and the number of \Envi tokens. The reasoning is detailed in Section~\ref{sec:end_proof}.

\subsection{Having more tokens makes things easier for \Envi}\label{sec:env cutoff}

First, we prove that beyond some threshold, if \Envi can win with some number of tokens, then they can \emph{a fortiori} win with a larger number of tokens. Let us stress that this is not true when the number of tokens is small, as witnessed by the formula \formule stating the existence of at least two \Envi tokens: in that case, \Sys can benefit from \Envi having more tokens.

\begin{lemma}
  \label{lem:env_cutoff}
  For every $k\in\N$, there exists some $\fE$ such that for any $\nS\in\N$, any $\nE\geq\fE$, and any \prefFOdw-sentence \formule of depth $k$, if $(\nS,\nE+1)$ is winning for \Sys then $(\nS,\nE)$ is winning for \Sys.
\end{lemma}

\begin{proof}
Let $k \in \N$ and let \formule be a \prefFOdw-sentence of depth $k$.
Let us fix three token sets $\TokS^{\nS}$, $\TokE^{\nE}$, and $\TokE^{\nE+1}$ of sizes \nS, \nE, \nE+1 respectively. We note $\Tok = \TokS^{\nS} \uplus \TokE^{\nE}$ and $\Tok_+ = \TokS^{\nS} \uplus \TokE^{\nE+1}$.
Let $\game = (\Tok, \arena_k, \Acc)$ be the token game over \Tok for \formule and $\game_+ = (\Tok_+, \arena_k, \Acc)$ the same over $\Tok_+$.
We show how to build a winning strategy for \Sys in \game from a winning strategy in $\game_+$.
But first, let us define some useful properties.

For all types $\type \in \Types$ we define the height of \type, denoted by \height, as its height in the tree of types, e.g. $\height[\type] = 0$ for any leaf in the tree.
Let \heightmax denote the height of the root $\type_0 = \tpword[\motvide]$.

In any given configuration obtained by following a play in \game, we say that a type \type is \emph{large} in that configuration if there are at least $k$ \Envi tokens in \type. 
Intuitively, this means that the acceptance condition $\Acc$ cannot differentiate between a configuration with a large type \type and the same configuration with even more tokens in \type.
Therefore, if we can ensure that \Sys has a strategy in \game that simulates the $\game_+$ winning strategy while always keeping the missing \Envi token in a large type, then that strategy would also be winning as $\Acc$ (which is the same in both \game and $\game_+$) has no way of distinguishing them.

To that end, we define what it means to have a \emph{huge} number of \Envi tokens in one type \type.
This is given by a lower bound $\boundtp = k \cdot |\Types|^{\height[\type]}$ that depends only on $k$ and the height of \type.
It guarantees the following properties:
\begin{enumerate}
\item A type that is huge is \textit{a fortiori} large.
\item If \type is a huge type with $\height[\type] = 0$, it contains at least $\boundtp = k$ \Envi tokens.
And since it is a leaf, all those tokens will stay in this type forever.
Thus \type will remain large from this point on.
\item If \type is a huge type with height greater than 0, after any \Envi move either \type still remains large, or by the pigeonhole principle there exists another type \typebis whose height is strictly lower than \height[\type], that can be reached from \type (i.e. such that \type \preordtp \typebis), and such that the number of \Envi tokens in \typebis is greater than \boundtp[\typebis], ensuring \typebis is also huge.
\end{enumerate}
We then define $\fE = \boundtp[\type_0] = k \cdot |\Types|^{\heightmax}$. Note that it depends only on $k$.

By these definitions, and since we assume that $\nE \geq \fE$, $\type_0$ is huge in the initial configuration because all \Envi tokens start in type $\type_0$.
By the third property, this means that after any move by \Envi, either $\type_0$ is still large, or there is (at least) one huge type $\type_1$ reachable from $\type_0$.
This can be repeated until either a type of height 0 is reached, which will be large forever according to the second property, or we stay in the same large type forever.
Formally, we define inductively a function $\largetype: \Plays \to \Types$ (\largetype for ``large type'') such that 
$\largetype(\conf_0) = \type_0$, 
$\largetype(\play \xrightarrow{\moveS} \conf) = \largetype(\play)$, and 
$\largetype(\play \xrightarrow{\moveE} \conf) = \type$ where \type is either a minimal (in terms of height) type such that $\largetype(\play) \preordtp \type$ and \type is huge in \conf if such a type exists, or $\type = \largetype(\play)$ otherwise. 
This is well-defined due to the above-mentioned third property, and we easily obtain that $\largetype(\play)$ is large in the last configuration of \play for any play \play.

Now assume \stratnEplus is a winning strategy for \Sys in $\game_+$.
We define a strategy \stratnE for \Sys in \game using $\largetype$ and additionally maintaining a play \playplus of $\game_+$ with the following invariant:
for all plays \play of \game that are \stratnE-compatible, \playplus is a \stratnEplus-compatible play such that the configuration reached after \playplus has the same number of tokens in each type as the configuration reached after \play, plus one extra \Envi token in $\largetype(\play)$.
The strategy \stratnE is simply defined as $\stratnE(\play) = \stratnEplus(\playplus)$, i.e. it mimics the actions of \stratnEplus on play \playplus.
Assuming that \playplus is properly defined and that the previously mentioned invariant holds, it is then easy to prove that \stratnE is winning.
Indeed, for every configuration that can be reached from a \stratnE-compatible play \play, there is an almost similar configuration that can be reached by following \playplus, which is \stratnEplus-compatible, with the only difference being one extra \Envi token in the type designated by \largetype.
Since this type is large by definition of \largetype, $\Acc$ cannot distinguish between the two configurations.
Thus, for any maximal play \play compatible with \stratnE, its limit configuration is indistinguishable from the limit configuration of \playplus, which satisfies \Acc by assumption of \stratnEplus being winning.
This proves that \stratnE is indeed a winning strategy for \Sys in \game.

It only remains to explain how \playplus is defined and show it satisfies the invariant.
Without loss of generality, assume that $\TokE^{\nE+1} = \TokE^{\nE} \uplus \{\ghost[]\}$.
We strengthen the second part of the invariant so that the configuration reached after \playplus is such that every token in $\TokE^{\nE}$ is in the same type as in the configuration reached after \play, and with \ghost[] being the extra token in type $\largetype(\play)$.
Initially this play \playplus is the empty play $\conf_0$ of $\game_+$, which trivially satisfies both conditions.
Suppose now that \play is a \stratnE-compatible play and \playplus is a \stratnEplus-compatible play that also satisfies the (strengthened) second part of the invariant.
\begin{itemize}
\item On any \Sys move $\moveS$ in \game leading to play $\play \xrightarrow{\moveS} \conf$, we simply update \playplus to $\playplus \xrightarrow{\moveS} \conf_+$ where $\conf_+$ is the result of applying $\moveS$ to the last configuration of $\playplus$.
Since \stratnE mimics \stratnEplus, if $\play \xrightarrow{\moveS} \conf$ is \stratnE-compatible, then $\playplus \xrightarrow{\moveS} \conf_+$ is \stratnEplus-compatible.
Moreover, by induction hypothesis, the last configuration of $\playplus$ is the last configuration of $\play$ with the extra token $\ghost[]$ in $\largetype(\play)$.
By applying the same \Sys move $\moveS$ to both, we easily obtain that $\conf_+$ is the same as $\conf$ plus $\ghost[]$ in $\largetype(\play \xrightarrow{\moveS} \conf) = \largetype(\play)$.
\item On an \Envi move $\moveE$ in \game resulting in play $\play \xrightarrow{\moveE} \conf$, let $\type = \largetype(\play)$ and $\typebis = \largetype(\play \xrightarrow{\moveE} \conf)$.
Let $\moveE^{\ghost[]}$ be the \Envi move that only affects \ghost[] by moving it from \type to \typebis (the move can be empty if \type = \typebis).
We know such a move is possible because by definition of \largetype we have that $\type \preordtp \typebis$.
Then with $\moveE^+$ being the \Envi move combining $\moveE$ and $\moveE^{\ghost[]}$, we update \playplus to $\playplus \xrightarrow{\moveE^+} \conf_+$.
It is trivially still \stratnEplus-compatible since no \Sys move has been made.
The extra token \ghost[] was moved to the new large type given by \largetype, and all other \Envi tokens made the same moves as in $\play \xrightarrow{\moveE} \conf$, so the invariant still holds.
\end{itemize}
Thus \playplus is properly defined and always satisfy the required invariant, which concludes the proof.
\ifcsname fullversion\endcsname
\else
\qed 
\fi
\end{proof}

\subsection{Too many tokens are useless to \Sys}\label{sec:sys_cutoff}

Dually, 
\ifcsname fullversion\endcsname
we prove in the appendix
\else
one can prove
\fi
the following lemma, stating that only so many tokens can help the \Sys win; beyond that point, no amount of additional tokens can turn the table and change a losing game into a winning one.

\begin{lemmarep}
  \label{lem:sys_cutoff}
  For every $k,\nE\in\N$, there exists some $\fS$ such that for any $\nS\geq\fS$, if \Sys has an $(\nS+1,\nE)$-winning strategy in a token game of depth $k$, then \Sys has an $(\nS,\nE)$-winning strategy in that game.
\end{lemmarep}

Note that contrary to Lemma~\ref{lem:env_cutoff}, where the bound depends only on $k$, here \fS depends both on $k$ and the number of tokens of \Envi. This cannot be avoided, as showcased by the following example.

\begin{example}
  Remember formula \exformule from Example~\ref{ex:main_ex}. 
  We can show that the $(\nS,\nE)$-game for \exformule cannot be won by \Sys when $\nS<\nE$, as exemplified in Figure~\ref{fig:ex_env_win}. Note that we use for simplicity's sake the representation of \Types[2] from Figure~\ref{fig:types} when we should consider \Types[3], as \exformule has depth $3$ -- this does not matter as \exformule cannot distinguish \preffotpword[\letterE\letterE][2] from \preffotpword[\letterE\letterE\letterE][2].

  Starting from the initial configuration depicted in Figure~\ref{fig:init}, \Envi can move one of their tokens from \preffotpword[\motvide][2] to \preffotpword[\letterE][2] (Figure~\ref{fig:firstE}), forcing \Sys to answer by moving on of their tokens from \preffotpword[\motvide][2] to \preffotpword[\letterS][2] (Figure~\ref{fig:firstS}) in order to satisfy the win condition of \exformule. \Sys could also directly move down other tokens from \preffotpword[\motvide][2] to \preffotpword[\letterS\letterS][2], but this would only make things worse.

  Then by moving the same token to \preffotpword[\letterE\letterE][2] as in Figure~\ref{fig:secondE}, \Envi would force \Sys to move as well their first token to \preffotpword[\letterS\letterS][2] (cf. Figure \ref{fig:secondS}). Repeating this sequence a total of \nS times on different tokens would end up ``using'' all of \Sys's tokens, which would all end up in \preffotpword[\letterS\letterS][2], as illustrated in Figure~\ref{fig:thirdS}. \Envi then only needs to move one last token to \preffotpword[\letterE][2] to falsify the win conditions for \exformule.
\end{example}

  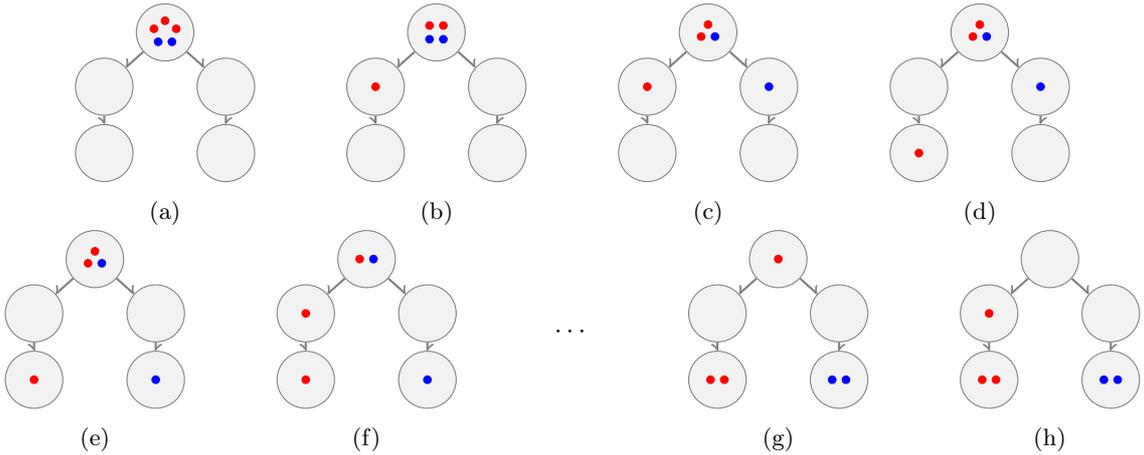
\begin{figure}[htb]
    \centering
    \captionsetup{justification=centering}
    \begin{subfigure}[b]{.2\textwidth}
      \centering
      \begin{tikzpicture}
        [scale=.8,
          every place/.style={draw=gray,fill=gray!10},
          every token/.style={draw=gray,fill=gray!10,minimum size=1mm}
        ]
        \node[place,
	  colored tokens={red,red,red,blue,blue}
        ] (vide) at (0,0) {};
        \node[place,
          colored tokens={}
        ] (aE) at (-1,-.9) {};
        \node[place,
          colored tokens={}
	] (aEaE) at (-1,-2) {};
        \node[place,
          colored tokens={}
	] (aS) at (1,-.9) {};
        \node[place,
	  colored tokens={}
        ] (aSaS) at (1,-2) {};
        \foreach \s/\t/\dir in {vide/aE/right to,aE/aEaE/right to,vide/aS/left to,aS/aSaS/left to} 
        \draw (\s) edge[gray,-\dir,thick] (\t);
      \end{tikzpicture}
      \caption{}
      \label{fig:init}
    \end{subfigure}
    \begin{subfigure}[b]{.2\textwidth}
      \centering
      \begin{tikzpicture}
        [scale=.8,
          every place/.style={draw=gray,fill=gray!10},
          every token/.style={draw=gray,fill=gray!10,minimum size=1mm}
        ]
        \node[place,
	  colored tokens={red,red,blue,blue}
        ] (vide) at (0,0) {};
        \node[place,
          colored tokens={red}
        ] (aE) at (-1,-.9) {};
        \node[place,
          colored tokens={}
	] (aEaE) at (-1,-2) {};
        \node[place,
          colored tokens={}
	] (aS) at (1,-.9) {};
        \node[place,
	  colored tokens={}
        ] (aSaS) at (1,-2) {};
        \foreach \s/\t/\dir in {vide/aE/right to,aE/aEaE/right to,vide/aS/left to,aS/aSaS/left to} 
        \draw (\s) edge[gray,-\dir,thick] (\t);
      \end{tikzpicture}
      \caption{}
      \label{fig:firstE}
    \end{subfigure}
    \begin{subfigure}[b]{.2\textwidth}
      \centering
      \begin{tikzpicture}
        [scale=.8,
          every place/.style={draw=gray,fill=gray!10},
          every token/.style={draw=gray,fill=gray!10,minimum size=1mm}
        ]
        \node[place,
	  colored tokens={red,red,blue}
        ] (vide) at (0,0) {};
        \node[place,
          colored tokens={red}
        ] (aE) at (-1,-.9) {};
        \node[place,
          colored tokens={}
	] (aEaE) at (-1,-2) {};
        \node[place,
          colored tokens={blue}
	] (aS) at (1,-.9) {};
        \node[place,
	  colored tokens={}
        ] (aSaS) at (1,-2) {};
        \foreach \s/\t/\dir in {vide/aE/right to,aE/aEaE/right to,vide/aS/left to,aS/aSaS/left to} 
        \draw (\s) edge[gray,-\dir,thick] (\t);
      \end{tikzpicture}
      \caption{}
      \label{fig:firstS}
    \end{subfigure}
    \begin{subfigure}[b]{.2\textwidth}
      \centering
      \begin{tikzpicture}
        [scale=.8,
          every place/.style={draw=gray,fill=gray!10},
          every token/.style={draw=gray,fill=gray!10,minimum size=1mm}
        ]
        \node[place,
	  colored tokens={red,red,blue}
        ] (vide) at (0,0) {};
        \node[place,
          colored tokens={}
        ] (aE) at (-1,-.9) {};
        \node[place,
          colored tokens={red}
	] (aEaE) at (-1,-2) {};
        \node[place,
          colored tokens={blue}
	] (aS) at (1,-.9) {};
        \node[place,
	  colored tokens={}
        ] (aSaS) at (1,-2) {};
        \foreach \s/\t/\dir in {vide/aE/right to,aE/aEaE/right to,vide/aS/left to,aS/aSaS/left to} 
        \draw (\s) edge[gray,-\dir,thick] (\t);
      \end{tikzpicture}
      \caption{}
      \label{fig:secondE}
    \end{subfigure}
    \begin{subfigure}[b]{.2\textwidth}
      \centering
      \begin{tikzpicture}
        [scale=.8,
          every place/.style={draw=gray,fill=gray!10},
          every token/.style={draw=gray,fill=gray!10,minimum size=1mm}
        ]
        \node[place,
	  colored tokens={red,red,blue}
        ] (vide) at (0,0) {};
        \node[place,
          colored tokens={}
        ] (aE) at (-1,-.9) {};
        \node[place,
          colored tokens={red}
	] (aEaE) at (-1,-2) {};
        \node[place,
          colored tokens={}
	] (aS) at (1,-.9) {};
        \node[place,
	  colored tokens={blue}
        ] (aSaS) at (1,-2) {};
        \foreach \s/\t/\dir in {vide/aE/right to,aE/aEaE/right to,vide/aS/left to,aS/aSaS/left to} 
        \draw (\s) edge[gray,-\dir,thick] (\t);
      \end{tikzpicture}
      \caption{}
      \label{fig:secondS}
    \end{subfigure}
    \begin{subfigure}[b]{.2\textwidth}
      \centering
      \begin{tikzpicture}
        [scale=.8,
          every place/.style={draw=gray,fill=gray!10},
          every token/.style={draw=gray,fill=gray!10,minimum size=1mm}
        ]
        \node[place,
	  colored tokens={red,blue}
        ] (vide) at (0,0) {};
        \node[place,
          colored tokens={red}
        ] (aE) at (-1,-.9) {};
        \node[place,
          colored tokens={red}
	] (aEaE) at (-1,-2) {};
        \node[place,
          colored tokens={}
	] (aS) at (1,-.9) {};
        \node[place,
	  colored tokens={blue}
        ] (aSaS) at (1,-2) {};
        \foreach \s/\t/\dir in {vide/aE/right to,aE/aEaE/right to,vide/aS/left to,aS/aSaS/left to} 
        \draw (\s) edge[gray,-\dir,thick] (\t);
      \end{tikzpicture}
      \caption{}
      \label{fig:thirdE}
    \end{subfigure}
    \begin{subfigure}[t]{.1\textwidth}
      \centering
      \begin{tikzpicture}
        \node at (0,0) {};
        \node at (0,1.5) {$\cdots$};
      \end{tikzpicture}
    \end{subfigure}
    \begin{subfigure}[b]{.2\textwidth}
      \centering
      \begin{tikzpicture}
        [scale=.8,
          every place/.style={draw=gray,fill=gray!10},
          every token/.style={draw=gray,fill=gray!10,minimum size=1mm}
        ]
        \node[place,
	  colored tokens={red}
        ] (vide) at (0,0) {};
        \node[place,
          colored tokens={}
        ] (aE) at (-1,-.9) {};
        \node[place,
          colored tokens={red,red}
	] (aEaE) at (-1,-2) {};
        \node[place,
          colored tokens={}
	] (aS) at (1,-.9) {};
        \node[place,
	  colored tokens={blue,blue}
        ] (aSaS) at (1,-2) {};
        \foreach \s/\t/\dir in {vide/aE/right to,aE/aEaE/right to,vide/aS/left to,aS/aSaS/left to} 
        \draw (\s) edge[gray,-\dir,thick] (\t);
      \end{tikzpicture}
      \caption{}
      \label{fig:thirdS}
    \end{subfigure}
    \begin{subfigure}[b]{.2\textwidth}
      \centering
      \begin{tikzpicture}
        [scale=.8,
          every place/.style={draw=gray,fill=gray!10},
          every token/.style={draw=gray,fill=gray!10,minimum size=1mm}
        ]
        \node[place,
	  colored tokens={}
        ] (vide) at (0,0) {};
        \node[place,
          colored tokens={red}
        ] (aE) at (-1,-.9) {};
        \node[place,
          colored tokens={red,red}
	] (aEaE) at (-1,-2) {};
        \node[place,
          colored tokens={}
	] (aS) at (1,-.9) {};
        \node[place,
	  colored tokens={blue,blue}
        ] (aSaS) at (1,-2) {};
        \foreach \s/\t/\dir in {vide/aE/right to,aE/aEaE/right to,vide/aS/left to,aS/aSaS/left to} 
        \draw (\s) edge[gray,-\dir,thick] (\t);
      \end{tikzpicture}
      \caption{}
      \label{fig:fourthE}
    \end{subfigure}
    \caption{\Sys may need at least as many tokens as \Envi to win.}
    \label{fig:ex_env_win}
  \end{figure}

  Although the insight gained in the proof of Lemma~\ref{lem:env_cutoff} is useful when considering the proof of Lemma~\ref{lem:sys_cutoff}, the latter is much more involved. Let us nevertheless give the key ideas of the proof. This time, the goal is to convert an $(\nS+1,\nE)$-winning strategy \oldstrat for \Sys in a token game of depth $k$ to an $(\nS,\nE)$-winning strategy \newstrat for the same game.
  
  In order to adapt \oldstrat to the situation where \Sys has one less token, we will track a \emph{ghost} token. The central idea of the proof is to guarantee at all time that the ghost shares its type with many other tokens, so that its presence or absence is of no import to the winning conditions of the game.
  Up to that point, the proof is very similar to that of Lemma~\ref{lem:env_cutoff}. The main difference in this case is that the ghost is not a fixed token: its identity among the \nS+1 \Sys tokens may vary depending on the way the play unfolds.
  But once again, if one starts with enough \Sys tokens and is careful in the tracking of this ghost token, it is possible to ``hide it'' among many others throughout the entire play.

\begin{toappendix}
  \ifcsname fullversion\endcsname
  \subsection{Idea of the proof and definitions.}
  \else
  \subsubsection{Idea of the proof and definitions.}
  \fi
Let us start by giving a value to \fS. For that, we define the following function $F$ by induction on $n$:
\begin{equation}
  \label{eq:F}
  \begin{cases}
    F(-1):=k\\
    F(n):= |\Types|^{\nE+2}\cdot\big(F(n-1)+1\big)^2+k &\text{ for }n\geq 0
  \end{cases}
\end{equation}
Let us note already that $F$ is non-decreasing, and that $F(n)\geq k$ for all $n > 0$. 
Let $\height[\type]$ denote the height of type \type in $(\Types,\preordtp)$, as defined in Section~\ref{sec:env cutoff}.
In a given configuration, we say that a type \type is \emph{large} if it has more than $k$ \Sys tokens and \emph{huge} if it has more than $F(\height[\type])$ \Sys tokens.
We then set
\begin{equation}
  \label{eq:fS}
  \fS:=\bound[\height[\preffotpword[\motvide]]]
\end{equation}
The relevance of this choice of $\fS$ will become apparent throughout the proof.

In order to prove Lemma~\ref{lem:sys_cutoff}, we explain how an $(\nS+1,\nE)$-winning strategy \oldstrat for \Sys in a token game of depth $k$ can be converted to an $(\nS,\nE)$-winning strategy \newstrat for the same game.

Let \oldTokS (resp. \newTokS) be the set of \Sys's tokens in the original $(\nS+1,\nE)$ (resp. new $(\nS,\nE)$) game, and $\TokE$ be the set of \Envi's tokens in both games. Note that $|\oldTokS|=\nS+1$, $|\newTokS|=\nS$ and $|\TokE|=\nE$.
Let $\gameplus = (\oldTokS \uplus \TokE, \arena_k, \Acc)$ and $\game = (\newTokS \uplus \TokE, \arena_k, \Acc)$ denote the corresponding games.

In order to adapt \oldstrat to a setting where \Sys has one fewer token, we will need to track a \emph{ghost} token \ghost; i.e. a \Sys token which appears in \oldstrat but not in \newstrat. The key idea is to make sure that at any point during the unfolding of \oldstrat, \ghost is always in a large type, so that the winning conditions of the game cannot notice its absence.
Contrary to the previous section, there is no hope for one token to be the ghost throughout the whole play: depending on the way \oldstrat unfolds, we may have to change candidates for being the ghost.

We will need additional information to decide when to switch to another token for our choice of ghost. First, we not only need to track the token that has disappeared when going from \oldstrat to \newstrat, but is it also crucial that we remember which token in \game plays the role of which token in \gameplus. 
We store this information in a one-to-one mapping $\map:\newTokS\to\oldTokS$, where \ghost is precisely the token from \oldTokS that has no preimage. 
Second, the strategy \newstrat we are defining mimics the winning strategy \oldstrat in \gameplus, and we need a way to remember how both connect. 
This information is stored via \playplus, such that each play \play in \game compatible with \newstrat has a corresponding play \playplus in \gameplus that is compatible with \oldstrat. 
Note that \playplus can be longer than \play, as we will sometimes need to move faster in \gameplus than in \game, for technical reasons.

It will be convenient to write \typeghost for $\playplus(\ghost)$, i.e. the type of the ghost after the play \playplus in \gameplus.

We work by induction on the length of the play \play, and define as we go along the strategy \newstrat on those \play which end with an \Envi move. If 
\begin{equation}
  \label{eq:play}
  \play=\conf_0 \xrightarrow{\moveE^0} \conf_1 \xrightarrow{\moveS^1} \conf_2 \xrightarrow{\moveE^2}\cdots\xrightarrow{m_?^{n-1}} \conf_{n}
\end{equation}
(where the last move is an \Envi move if and only if $n$ is odd), we denote \partplay the prefix $\play=\conf_0 \xrightarrow{\moveE^0} \conf_1 \xrightarrow{\moveS^1} \cdots \xrightarrow{m_{?}^{i-1}}\conf_{i}$ of \play.

Let us consider a play \play as in (\ref{eq:play}) compatible with (what has been defined so far of) \newstrat, which satisfies the invariants defined below. Then 
\begin{itemize}
\item if $n=2l$ is even, we show that for any valid \Envi move \nextMoveE compatible with \play and with $\playbis = \play \xrightarrow{\nextMoveE} \conf_{n+1}$, 
we can define \playbisplus and \map[\playbis] such that the invariants also hold for $\playbis$
\item if $n=2l+1$ is odd, we define the next \Sys move $\newstrat(\play)$ of our new strategy, and with $\playbis = \play \xrightarrow{\newstrat(\play)} \conf_{n+1}$ we also define \playbisplus and \map[\playbis] such that the invariants also hold for $\playbis$.
\end{itemize}
To conclude the proof of Lemma~\ref{lem:sys_cutoff} we show that the new strategy \newstrat is indeed winning in \game.

One of the invariant states that the type of the ghost should go down as much as possible when choosing the next \Sys move.
To that end, for any play \play ending in an \Envi move, we define a \emph{$\type$-descending sequence} for \play as a pair $(\perm,\seqmoves)$ such that \perm is a permutation of \oldTokS, \seqmoves is a finite sequence of moves in \gameplus starting from \playplus that starts and end with a \Sys move, that satisfies the following conditions:
\begin{align}
\label{eq:type}
\begin{cases}
  \playbisplus = \playplus \xrightarrow{\seqmoves} \conf'_+ \text{ is compatible with }\oldstrat\\
  \forall\pe\in\TokE,\ \playplus(\pe) = \playbisplus(\pe)\\
  \forall\ps\in\oldTokS,\ \playplus(\ps) \preordtp \playbisplus(\perm(\ps))\\
  \playbisplus(\perm(\ghost[\play]))=\type\\ 
  |\{\ps\in\oldTokS:\playbisplus(\ps)=\type\}|\geq\bound+1
\end{cases}
\end{align}
A \Sys move $\moveS$ is said to be \emph{max-descending} if it results in a play $\playbis = \play \xrightarrow{\moveS} \conf'$ such that there are no $\type$-descending sequence for \play for any type $\type \neq \typeghost[\playbis]$ such that $\typeghost[\playbis] \preordtp \type$.

  \ifcsname fullversion\endcsname
  \subsection{Statement of the invariants.}
  \else
  \subsubsection{Statement of the invariants.}
  \fi

We are now ready to list the various invariants, which are indexed with the play to which they refer. We give below an intuition about their meaning and usefulness.

\begin{tabular}{rl}
  \invnu:&\playplus is compatible with \oldstrat and \playplus has the same length parity as \play\\
  &furthermore, for every prefix \playbis of \play, \playbisplus is a prefix of \playplus\\
  \invetp:&$\forall \pe\in\TokE,\ \play(\pe) = \playplus(\pe)$ \\
  \invstp:&$\forall \ps\in\newTokS,\ \play(\ps)=\playplus(\map(\ps))$\\
  \invghosttp:& If \playbis is a prefix of \play, then $\typeghost[\playbis]\preordtp \typeghost$\\
  \invbound:&If there is no strict prefix \playbis of \play such that $\typeghost[\playbis]=\typeghost$, \\
  &then \typeghost is huge in \play; otherwise \typeghost is still large in \play.\\
  \invmax:&All \Sys moves are max-descending.\\
\end{tabular}

Let us briefly give some intuition about these invariants. First of all, \invnu, \invetp and \invstp ensure that \play mimics \playplus. In particular, then all tokens (if we identify $\map(\ps)$ with $\ps$) except the ghost have the same type after \play and \playplus. The invariant \invghosttp ensures that, although the ghost can change during the play, the type containing it can only go down in \Types. Another crucial requirement, if we are to show that the limit configuration of \play in \game is winning if and only if the limit configuration of \oldplay is winning in \gameplus, is that the ghost (which is an extra token in \gameplus) cannot make a difference for the acceptance condition of the game. This is ensured by \invbound, which has as consequence the fact that the ghost is always in a large type. The first part of \invbound is a technical requirement needed in the proofs, and ensures that whenever the type of the ghost moves down in \Types, its new type is huge. As for \invmax, it ensures that the \Sys brings the ghost down as much as possible: if a play by \Sys would have permitted the ghost to be lower in the tree $(\Types,\preordtp)$, then the \Sys would have played it. This is needed to guarantee \invbound.

  \ifcsname fullversion\endcsname
  \subsection{Proof by induction.}
  \else
  \subsubsection{Proof by induction.}
  \fi

We are now ready to move to the induction proper.
\paragraph*{Empty play.}

In order to start the induction, let us consider the empty play $\play_0$ in \game. We set $\map[\play_0]$ to be any one-to-one mapping $\newTokS\to\oldTokS$, and $(\play_0)_+$ to be the empty play in \gameplus. One easily checks that all the invariants are satisfied for $\play_0$; in particular, \invbound[\play_0] follows from the choice of \fS in (\ref{eq:fS}) and from the assumption that $\nS\geq\fS$.

\paragraph*{When \Envi makes a move.}

Let \play be a play ending with a \Sys move, and consider the next move \nextMoveE by \Envi resulting in play \playbis. We set $\map[\playbis]:=\map$ and $\playbisplus := \playplus \xrightarrow{\nextMoveE} \conf_+$. The latter is well defined according to \invetp. All the invariants are trivially satisfied with these definitions.

\paragraph*{When \Sys makes a move.}

Let us now consider a play \play ending with an \Envi move. By induction hypothesis, we can assume that all the $(\partplay[2i])_+$ and $\map[\partplay[2i]]$, for $i\leq l$, are defined and satisfy all the invariants. 
We will now explain how to construct the next \Sys move $\newstrat(\play)$ resulting in play $\playbis = \play \xrightarrow{\newstrat(\play)} \conf_{n+1}$, 
as well as $\playbisplus$ and $\map[\playbis]$ so that these invariants still hold for $\playbis$.
We distinguish two cases: when the ghost token moves down in the tree $(\Types,\preordtp)$ and when it stays in the same type.

\paragraph*{When the ghost goes down in the tree.}

Let \type be a maximal type such that $\type \neq \typeghost$, $\typeghost \preordtp \type$, and there is a \type-descending sequence exists in \play with associated \perm and \seqmoves, assuming such a type exists.
Let $\playbisplus = \playplus \xrightarrow{\seqmoves} \conf'_+$.
In this case, we define the move $\newstrat(\play)$ as follows:
\begin{equation}
  \label{eq:def_strat}
  \forall \ps \in\newTokS,\ \newstrat(\play)(\ps):=\playbisplus(\perm(\map(\ps)))\,.
\end{equation}
This is well defined, insomuch as for every $\ps\in\newTokS$ we have, according to \invstp and (\ref{eq:type}), 
\[\play(\ps) = \playplus(\map(\ps)) \preordtp \playbisplus(\perm(\map(\ps)))\,.\]
Let us write \nextplay for the play \play followed by our newly defined response:
\[\nextplay:=\play \xrightarrow{\newstrat(\play)} \conf'\,.\]
Let us now define $\nextplayplus$ and \nextmap as follows:
\begin{equation}
  \label{eq:def_nu}
  \nextplayplus:=\playbisplus
\end{equation}
and
\begin{equation}
  \label{eq:def_map}
  \forall\ps\in\newTokS,\ \nextmap(\ps):=\perm(\map(\ps))\,.
\end{equation}
Notice that these definitions entail
\begin{equation}
  \label{eq:def_ghost}
  \nextghost=\perm(\ghost)\quad\text{and}\quad\nexttypeghost=\type\,.
\end{equation}
Let us show that with these definitions, the invariants hold on \nextplay:
\begin{description}
\item[Proof of \nextinvnu:] Because \seqmoves satisfies (\ref{eq:type}), $\nextplayplus = \playplus \xrightarrow{\seqmoves} \conf'_+$ is compatible with \oldstrat. It ends with a \Sys move, as does \nextplay.
The second part of the invariant comes trivially as \nextplayplus is built by extending \playplus.

\item[Proof of \nextinvetp:] We need to show that for every $\pe\in\TokE$, $\nextplay(\pe) = \nextplayplus(\pe)$. Because \seqmoves satisifes (\ref{eq:type}), we have $\playplus(\pe) = \playbisplus(\pe)$, from which \nextinvetp follows, since for every $\pe\in\TokE$, $\playplus(\pe)=\nextplay(\pe)$ (as a \Sys move doesn't change the type of an \Envi token) and $\nextplayplus=\playbisplus$ (\ref{eq:def_nu}).

\item[Proof of \nextinvstp:] Let us show that for every $\ps\in\newTokS$, $\nextplay(\ps)=\nextplayplus(\nextmap(\ps))$.
By definition, $\nextplay(\ps)=\playbisplus(\perm(\map(\ps)))$ which, by definition of \nextplayplus (\ref{eq:def_nu}) and \nextmap (\ref{eq:def_map}), amounts to $\nextplayplus(\nextmap(\ps))$.

\item[Proof of \nextinvghosttp:]
This follows directly from $\typeghost\preordtp\type=\nexttypeghost$ and \invghosttp.

\item[Proof of \nextinvbound:]
Let us show that \nexttypeghost is huge in \nextplay.
By (\ref{eq:type}) and (\ref{eq:def_ghost}), we have $|\{\ps\in\oldTokS:\nextplayplus(\ps)=\nexttypeghost\}|\geq\bound[\height[\nexttypeghost]]+1$. It follows from \nextinvstp that $|\{\ps\in\newTokS:\nextplay(\ps)=\nexttypeghost\}|\geq\bound[\height[\nexttypeghost]]$, as every token of \oldTokS but one (the ghost) has a preimage by \nextmap.

\item[Proof of \nextinvmax:] The maximality in our choice of \type ensures \nextinvmax holds.
\end{description}

\paragraph*{When the ghost stays put.}
Assume now that there is no type \type such that $\type \neq \typeghost$, $\typeghost \preordtp \type$, and there is a $\type$-descending sequence in \play.
We then prove the following:

\begin{claim}
  There exists at least $k+1$ tokens $s\in\oldTokS$ such that $\oldstrat(\playplus)(s)=\typeghost$.
\end{claim}

\begin{proof}
  Let us decompose the play \play as \[\play=\conf_0 \xrightarrow{\moveE^0} \conf_1 \xrightarrow{\moveS^1} \conf_2 \xrightarrow{\moveE^2}\cdots\xrightarrow{\moveS^{2l-1}} \conf_{2l}\xrightarrow{\moveE^{2l}} \conf_{2l+1}\]
  It will be convenient to extend the notation of partial plays and to write $(\partplay[2l+2])_+$ to denote the play $\playplus \xrightarrow{\oldstrat(\playplus)} \conf'_+$.
  
  According to \invghosttp, the place of the ghost can only go down during the play \playplus in \gameplus. 
Let us consider the sequence $i_1<\cdots <i_n$ of all indices $i\leq l$ such that
  \begin{itemize}
  \item $\typeghost[\partplay[2i+1]]=\typeghost$, meaning that the ghost at that point was already in the same place as it is after the play \play, and
  \item if we denote by $N_i$ the number of tokens $\ps\in\oldTokS$ such that \[(\partplay[2i+1])_+(\ps) = \typeghost \text{ and } (\partplay[2i+2])_+(\ps)\neq \typeghost\,,\]
    we have $N_i\geq 1$. In other words, we consider only partial plays such that in the next move, at least one \Sys token moved down from \typeghost.
  \end{itemize}

  Because of \invbound and \invstp, we know that after $(\partplay[2{i_1}+1])_+$, the number of \Sys tokens in \typeghost is at least $\bound[\height[\typeghost]]+1$. Thus
  \begin{equation}
    \label{eq:fuite}
    \sum_{i\in\{i_1,\dots,i_n\}}N_i+|\{s\in\oldTokS:\oldstrat(\playplus)(s)=\typeghost\}|\geq \bound[\height[\typeghost]]+1\,.
  \end{equation}
  We will show that $\sum_{i\in\{i_1,\dots,i_n\}}N_i$ cannot be too large, by proving that for every $i$,
  \begin{equation}
    \label{eq:small_N}
    N_i\leq|\Types|\cdot \big(\bound[\height[\typeghost]-1]+1\big)\,,
  \end{equation}
  and that $n$ also must remain rather small:
  \begin{equation}
    \label{eq:small_n}
    n\leq|\Types|^{\nE+1}\cdot(\bound[\height[\typeghost]-1]+1)\,.
  \end{equation}
  Combining the definition of $F$ in (\ref{eq:F}) with (\ref{eq:fuite}), (\ref{eq:small_N}) and (\ref{eq:small_n}) we get \[|\{s\in\oldTokS:\playplus(s)=\typeghost\}|\geq k+1\] as stated by the claim. Note that one can prove tighter bounds for (\ref{eq:small_N}) and (\ref{eq:small_n}), and thus pick a smaller \fS in (\ref{eq:fS}). We refrain from doing so for the sake of readability.

  Let us first show that (\ref{eq:small_N}) holds: if there exists some $i\in\{i_1,\dots,i_n\}$ such that \[N_i> |\Types|\cdot \big(\bound[\height[\typeghost]-1]+1\big)\] then by the pigeonhole principle, at least one type $\type\neq\typeghost$ with $\typeghost\preordtp\type$ is such that there exist at least $\bound[\height[\typeghost]-1]+1$ (which is at least $\bound+1$ because $F$ is non-decreasing) tokens $s\in\oldTokS$ such that 
  \[(\partplay[2i+1])_+(\ps) = \typeghost = \typeghost[\partplay[2i+1]] \text{ and } (\partplay[2i+2])_+(\ps)= \type\,.\] 
  Let us consider one such token $\oldps\in\oldTokS$. With $\perm:=(\oldps\,\ghost[2i+1])$ (the permutation swapping \oldps and \ghost[2i+1]) and $\seqmoves:=\oldstrat((\partplay[2i+1])_+)$, this \type contradicts \invmax[\partplay[2i+2]] -- if $i=l$, it contradicts the non-existence of a \type with a \type-descending sequence.

  To conclude the proof, let us prove that (\ref{eq:small_n}) holds. Assume toward a contradiction that \[n\geq |\Types|^{\nE+1}\cdot(\bound[\height[\typeghost]-1]+1)\,.\] By the pigeonhole principle, there must exist a subsequence $j_1<\cdots<j_{n'}$ of $i_1<\cdots<i_n$ of length $n'\geq|\Types|\cdot(\bound[\height[\typeghost]-1]+1)$ such that for $j,j'\in\{j_1,\dots,j_{n'}\}$, 
  \[\forall\pe\in\TokE,\ (\partplay[2j+1])_+(\pe)=(\partplay[2j'+1])_+(\pe)\,.\]
  Given that $N_j\geq 1$ for each such $j$, the pigeonhole principle guarantees the existence of some $\type\neq\typeghost=\typeghost[2j_1+2]$ such that $\typeghost[2j_1+1]\preordtp\type$ and 
  \[\exists^{\geq\bound[\height[\typeghost]+1]+1}\ps\in\oldTokS,\ (\partplay[2j_1+1])_+(\ps) = \typeghost \text{ and } (\partplay[2j_{n'}+2])_+(\ps)=\type\,.\] 
  Let us consider one such token $\oldps\in\oldTokS$. With $\perm:=(\oldps\,\ghost[2j_1+1])$ (the permutation swapping \oldps and \ghost[2j_1+1]) and \[\seqmoves:=\moveS^{2j_1+1}\cdot\moveE^{2j_1+2}\cdots\moveS^{2j_{n'}+2}\,,\] this \type-descending sequence contradicts \invmax[\partplay[2j_1+2]], which ends the proof of the claim.
  \ifcsname fullversion\endcsname
  \else
  \qed
  \fi
\end{proof}

Let us consider a token $\oldps\in\oldTokS\cap\text{Im}(\map)$ such that $\oldstrat(\playplus)(\oldps)=\typeghost$ and the token $\newps\in\newTokS$ such that $\map(\newps)=\oldps$.

We define $\newstrat(\play)$ as follows:
\begin{equation}
  \label{eq:def_strat_static}
  \begin{cases}
    \forall \ps\in\newTokS\setminus\{\newps\},\ \newstrat(\play)(\ps):=\oldstrat(\playplus)(\map(\ps))\\
    \newstrat(\play)(\newps):=\newstrat(\playplus)(\ghost)
  \end{cases}
\end{equation}
This is well defined, as for every $\ps\in\newTokS$ different from \newps we have 
\[\play(\ps)=\playplus(\map(\ps))\preordtp\oldstrat(\playplus)(\map(\ps))\] 
as well as
 $\play(\newps)=\playplus(\oldps)\preordtp\typeghost\preordtp\newstrat(\playplus)(\ghost)$.

Once again, we write \nextplay for the play \play followed by the response we just defined
\[\nextplay:=\play \xrightarrow{\newstrat(\play)} \conf'\,.\]
We then define
\begin{equation}
  \label{eq:def_nu_static}
  \nextplayplus:=\playplus \xrightarrow{\oldstrat(\playplus)} \conf'_+
\end{equation}
and
\begin{equation}
  \label{eq:def_map_static}
  \forall\ps\in\newTokS,\ \nextmap(\ps):=
  \begin{cases}
    \ghost&\text{if }\ps=\newps\\
    \map(\ps)&\text{otherwise}
  \end{cases}
\end{equation}
which immediately gives
\begin{equation}
  \label{eq:def_ghost_static}
  \nextghost=\oldps\quad\text{and}\quad\nexttypeghost=\typeghost\,.
\end{equation}

Invariant \nextinvbound is satisfied due to the previous claim, \nextinvmax due to the assumption that there is no descending sequence, and all other invariants are straightforward from the definition of \nextplay.

  \ifcsname fullversion\endcsname
  \subsection{\newstrat is a winning strategy for \Sys.}
  \else
  \subsubsection{\newstrat is a winning strategy for \Sys.}
  \fi

In order to conclude the proof of Lemma~\ref{lem:sys_cutoff}, it remains to prove that \newstrat is a winning strategy for \Sys in \game. 
To that end, let us consider a maximal play $\widehat{\play}$ compatible with \newstrat. 
Let us define the play $\widehat{\play}_+$ defined as the limit of $\playplus$ for increasing prefixes \play of $\widehat{\play}$, which is well defined due to \invnu.

There exists a finite prefix \play of $\widehat{\play}$ such that the configurations reached after \play and \playplus are the limit configurations of $\widehat{\play}$ and $\widehat{\play}_+$ respectively.
Since $\widehat{\play}_+$ is a winning play for \Sys, the configuration reached after \playplus must be winning for \Sys.
Let us argue that the last configuration of \play is winning, which entails that $\widehat{\play}$ is a winning play for \Sys.

For every $\type\in\Types$,
\begin{itemize}
\item the number of \Envi tokens in \type after \play and the number of \Envi tokens in \type after \playplus are equal, in view of \invetp, and
\item following \invstp and \invbound, the number of \Sys tokens in \type after \play and the number of \Sys tokens in \type after \playplus are either equal or both larger than $k$, $k$ being the depth of the game.
\end{itemize}
Since the configuration reached after \playplus is winning for \Sys, the configuration reached after \play is winning as well for \Sys.

All in all, we have shown that any winning strategy \oldstrat for \Sys in \gameplus ensures the existence of a winning strategy \newstrat for \Sys in \game. This concludes the proof of Lemma~\ref{lem:sys_cutoff}.
\ifcsname fullversion\endcsname
\else
\qed
\fi
\end{toappendix}

\subsection{Solving the synthesis problem}\label{sec:end_proof}
Remember that our goal is to decide whether there exists a pair $(\nS,\nE) \in \N^2$ that is winning for a given formula \formule of depth $k$.
Using Lemmas~\ref{lem:env_cutoff} and \ref{lem:sys_cutoff}, we have shown that we can restrict the search space to the set $N = \{(\nS,\nE) \in \N^2 \mid \nE \leq \fE \land \nS \leq \fS\}$ (which is finite and computable) in the sense that if there is no winning pair in $N$, then there will be no winning pair in $\N^2$.

Recall that for a fixed pair $(\nS,\nE)$, the corresponding token game is a finite game with finite configurations.
As such, it can be solved by seeing it as a game on the graph of configurations, with a Büchi winning condition where accepting configurations are exactly those that satisfy \Acc.
From there, a winning strategy can be mapped back to a winning strategy in the $(\nS,\nE)$ token game by looking at the history of a play and constructing the configurations from it.
Finally, as the proofs of equivalence in Section~\ref{sec:synthesis} are constructive, one can then go from a winning strategy in the token game back to a winning strategy in the original synthesis game.
Therefore, the synthesis problem can be solved by iterating on all pairs in $N$ and solving the token game for that pair, and then accepting the first winning pair found if there is one, and rejecting if no such pair is found. 

\begin{figure}
\tikzset{every state/.style={minimum size=15pt}}
\begin{center}
  \begin{tikzpicture}[scale=0.9]
    \begin{scope}
    \draw[->,semithick] (0,0) -- (3,0) node [below] {$\nE$};
    \draw[->,semithick] (0,0) -- (0,3) node [left] {$\nS$};
    \foreach \x in {0,...,5}
      \foreach \y in {0,...,5}
        \draw (0.5*\x,0.5*\y) node {$\cdot$};
    \end{scope}
    \draw[->] (3.2,1.5) -- (4.3,1.5) node [above left] {\tiny Lemma~\ref{lem:env_cutoff}};
    \begin{scope}[shift={(4.5,0)}]
    \draw[->,semithick] (0,0) -- (3,0) node [below] {$\nE$};
    \draw[->,semithick] (0,0) -- (0,3) node [left] {$\nS$};
    \foreach \x in {0,...,5}
      \foreach \y in {0,...,5}
        \draw (0.5*\x,0.5*\y) node {$\cdot$};
    \draw (2,0) node [below] {$\fE$};
    \draw[color=red] (2,0) -- (2,3);
    \end{scope}
    \draw[->] (7.7,1.5) -- (8.8,1.5) node [above left] {\tiny Lemma~\ref{lem:sys_cutoff}};
    \begin{scope}[shift={(9,0)}]
    \draw[->,semithick] (0,0) -- (3,0) node [below] {$\nE$};
    \draw[->,semithick] (0,0) -- (0,3) node [left] {$\nS$};
    \fill[pattern=north west lines, pattern color=red!30] (0,0) -- (0,.5) -- (.5,.5) -- (1,1) -- (1.5,2) -- (1.75,3) -- (2,3) -- (2,0) --(0,0);
    \foreach \x in {0,...,5}
      \foreach \y in {0,...,5}
        \draw (0.5*\x,0.5*\y) node {$\cdot$};
    \draw (2,0) node [below] {$\fE$};
    \draw[color=red] (2,0) -- (2,3);
    \draw (0,.5) node[left] {$\fS$};
    \draw[color=red] (0,0.5) -- (0.5,0.5) -- (1,1) -- (1.5,2) -- (1.75,3);
    \draw (1.25,0.77) node {\textcolor{red}{$N$}};
    \end{scope}

  \end{tikzpicture}
\caption{Bounding the search space for winning pairs.}\label{fig:solving}
\end{center}
\end{figure}
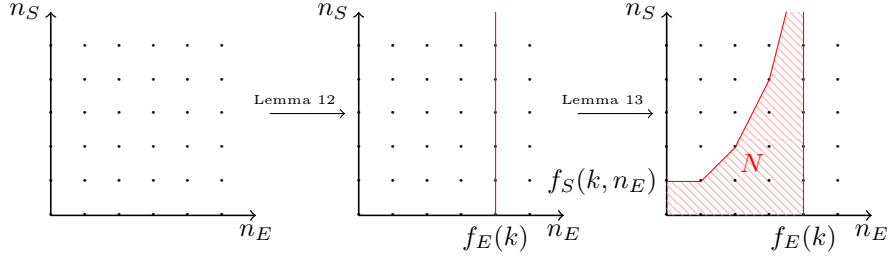

\begin{toappendix}
\section{Relation between \FO and \prefFO types}
In this section, we investigate the link between \FO and \prefFO types on regular words.
Note first that every \prefFO-sentence of quantifier depth $k$ is equivalent to some \FO-sentence of quantifier depth $k$. What this means is that \FO $k$-types are more fine-grained than \prefFO $k$-types -- in particular, it makes sense to speak of the \prefFO $k$-type of an \FO $k$-type.

Loosely speaking, two finite words $u$ and $v$ agree on \prefFO sentences if and only if there exists words $u'$ and $v'$ such that $u\cdot u'$ and $v$ agree on \FO sentences, as well as $u$ and $v\cdot v'$. This is stated more precisely in the two following lemmas.

\begin{lemma}
  Let $u, v$ be two finite words for which there exists $u', v'$ such that $u\cdot u'\foeq v$ and $u\foeq v\cdot v'$. Then $u\prefFOeq v$.
\end{lemma}

\begin{proof}
  Since $\foeq$ is a congruence for the concatenation in the monoid of finite words, the assumptions combine to give $u\cdot u'\cdot v'\foeq u$, which implies $u\cdot u'\cdot v'\prefFOeq u$.

  Lemma~\ref{lem:sandwich} then ensures that $u\cdot u'\prefFOeq u$. This equivalence, combined with $u\cdot u'\prefFOeq v$ (which follows directly from the assumptions), concludes the proof.
  \ifcsname fullversion\endcsname
  \else
  \qed
  \fi
\end{proof}

Conversely, we have:

\begin{lemma}
  Let $u, v$ be two finite words such that $u\prefFOeq[k+1] v$. Then there exists $u', v'$ such that $u\cdot u'\foeq v$ and $u\foeq v\cdot v'$.
\end{lemma}

Before turning to the proof of this lemma, let us note that the requirement that $u$ and $v$ be \prefFO similar at depth $k+1$ (and not just $k$) is necessary. Consider indeed the words \[u:=abaac\] and \[v:=abaaac\] in the case $k:=3$.

Let us consider the following strategy for the Duplicator in the $3$-round prefix \EF game between $u$ and $v$: if the Spoiler places their bounding pebble on any of the first four positions of $u$ or $v$, then the Duplicator responds by playing on the same position in the other word, and wins since both prefixes are then isomorphic. Otherwise, if the Spoiler starts by playing on the last $a$ of $v$, then the Duplicator plays on the last $a$ of $u$, and can easily win two more rounds from there. Finally, if the Spoiler makes their first move on a $c$, then the Duplicator responds on the other $c$, and can once again win two more rounds from that configuration. We thus have $u\prefFOeq[3]v$.

However, for any word $u'$, $u\cdot u'$ and $v$ will disagree on the \FO-sentence of quantifier depth $3$ stating the existence of three consecutives $a$'s before the first occurrence of $c$.

\begin{proof}
  Let $u$ and $v$ be finite words such that $u\prefFOeq[k+1] v$. We can assume that both of them have at least one letter, for otherwise we already have $u\foeq v$.

  We show the existence of some word $u'$ such that $u\cdot u'\foeq v$, and we conclude by symmetry. By assumption, the Duplicator has a winning strategy in the $(k+1)$-round prefix \EF game between $u$ and $v$. Let $n$ be the position in $v$ of the Duplicator's response when the Spoiler first plays on the last letter of $u$; from that configuration, the Duplicator can win for $k$ for rounds. We claim that $u\foeq v[0\dots n]$: indeed, in the $k$-round (regular) \EF game between $u$ and $v[0\dots n]$, the Duplicator can follow the above-mentioned strategy (in which every play in $v$ is left of the bounding pebble, i.e. in $v[0\dots n]$) and win.

  Let us now define $u':=v[n+1\dots]$: by virtue of $\foeq$ being a congruence for the concatenation of finite words, $u\foeq v[0\dots n]$ entails $u\cdot u'\foeq v$, which concludes the proof.
  \ifcsname fullversion\endcsname
  \else
  \qed
  \fi
\end{proof}
\end{toappendix}

\section{Conclusion}\label{sec:conclu}
We have shown that the synthesis problem is decidable for \prefFOdw, the prefix first-order logic on data words.
Our proof is based on successively bounding the number of Environment and System processes that are relevant for solving the problem, thus restricting the search space to a finite set.

The correctness of those bounds heavily relies on the nice properties exhibited by \prefFOdw.
The most important of them is that \prefFO-types on regular words form a tree, which would not be the case with the unrestricted $\FO[\lesssim]$.
We show in the
\ifcsname fullversion\endcsname
appendix
\else
full version of this paper
\fi
that \prefFOdw is actually the finest restriction enjoying this property, in the sense that \prefFO-types correspond exactly to connected components in the graph of types of $\FO[\lesssim]$.
To the best of our knowledge, this natural restriction and the properties it enjoys have not been studied anywhere else.

On the synthesis side, our result narrows the gap between decidability and undecidability of first-order logic fragments on data words.
Previous results had shown that adding unrestricted order to the allowed predicates immediately lead to undecidability, greatly reducing the kind of specifications that could be written.
It turns out that the limited kind of order added by the predicate $\lesssim$ does not share the same outcome, so we obtain a fragment with more expressivity than before and for which synthesis is still decidable.
The decidability for the full $\FO[\lesssim]$ remains open; our technique for tracking a missing ghost token cannot be easily adapted in that setting as the type structure for that logic is not as nice as in the \prefFOdw setting.
We conjecture that it remains decidable, and leave this as potential future works.

\bibliographystyle{splncs04}
\bibliography{biblio}

\end{document}